\documentclass[]{article}

\usepackage[utf8]{inputenc}
\usepackage[hidelinks]{hyperref}
\usepackage{url}
\usepackage{amsfonts}
\usepackage{microtype}
\usepackage{xcolor}
\usepackage{times}
\usepackage{authblk}
\usepackage{amsmath}
\usepackage{amssymb}
\usepackage{amsthm}
\usepackage{algorithm}
\usepackage{algorithmic}
\usepackage[round]{natbib}
\usepackage{tikz}

\usepackage{geometry}
\theoremstyle{plain}
\newtheorem{theorem}{Theorem}[section]

\newtheorem{lemma}[theorem]{Lemma}
\newtheorem{corollary}[theorem]{Corollary}
\newtheorem{fact}[theorem]{Fact}
\theoremstyle{definition}
\newtheorem{definition}[theorem]{Definition}

\theoremstyle{remark}
\newtheorem{remark}[theorem]{Remark}

\usepackage{tikz}
\usetikzlibrary{matrix}

\newcommand{\bernoulli}{\mathrm{Bernoulli}}
\newcommand{\uniform}{\mathrm{Uniform}}
\newcommand{\geometric}{\mathrm{Geometric}}

\newcommand{\multinomial}{\mathrm{Multinomial}}
\newcommand{\poisson}{\mathrm{Poisson}}
\newcommand{\eps}{\varepsilon}
\renewcommand{\epsilon}{\eps}
\newcommand{\E}{\mathbb{E}}
\newcommand{\M}{\mathcal M}
\newcommand{\F}{\mathcal F}
\newcommand{\R}{\mathbb R}

\newcommand{\X}{\mathcal X}
\newcommand{\scal}{\mathcal S}
\newcommand{\ical}{\mathcal I}
\newcommand{\fn}[1]{\texttt{#1}}
\newcommand{\defn}{\ensuremath{:=}}
\newcommand{\lxor}{\veebar} 
\newcommand{\Mxor}{\M^{\operatorname{xor}}}
\newcommand{\Msym}{\M^{\operatorname{sym}}}
\newcommand{\Mps}{\M^{\operatorname{PS}}}

\DeclareMathOperator{\merge}{merge}
\DeclareMathOperator{\Var}{Var}

\newcommand{\ussym}{\textit{SFM (sym)}}
\newcommand{\usxor}{\textit{SFM (xor)}}
\newcommand{\psloose}{\textit{PS (loose)}}
\newcommand{\pstight}{\textit{PS (tight)}}

\newcommand{\estse}{\widehat{\mathrm{SE}}}

\title{Sketch-Flip-Merge: Mergeable Sketches for Private Distinct Counting}

\author[1]{Jonathan Hehir\footnote{This research was conducted while the author was at Meta.}}
\author[2]{Daniel Ting}
\author[3]{Graham Cormode}
\affil[1]{Department of Statistics, Pennsylvania State University, USA}
\affil[2]{Meta, USA}
\affil[3]{Meta, UK}
\date{\today}

\begin{document}

\maketitle

\begin{abstract}
Data sketching is a critical tool for distinct counting, enabling multisets to be represented by compact summaries that admit fast cardinality estimates. Because sketches may be merged to summarize multiset unions, they are a basic building block in data warehouses. Although many practical sketches for cardinality estimation exist, none provide privacy when merging. We propose the first practical cardinality sketches that are simultaneously mergeable, differentially private (DP), and have low empirical errors. These introduce a novel randomized algorithm for performing logical operations on noisy bits, a tight privacy analysis, and provably optimal estimation. Our sketches dramatically outperform existing theoretical solutions in simulations and on real-world data.
\end{abstract}

\section{Introduction}
Many applications that model large volumes of data are based on tracking cardinalities of events or observations.   
Consequently, these applications make extensive use of data sketches that support fast, approximate cardinality estimation~\citep{cormode2020small}. 
For instance, approximate distinct counting is supported via variants of the HyperLogLog (HLL) sketch \citep{flajolet2007hyperloglog,heule2013hyperloglog} in popular data management systems including Amazon Redshift, ClickHouse, Google BigQuery, Splunk, Presto, Redis, and more. 
At the expense of a small estimation error, these approximate methods drastically reduce the computational cost of distinct counting to run in linear time, using only bounded memory. 
An additional key feature of distinct-count sketches is the ability to merge two or more sketches to obtain cardinality estimates over their union. This enables not only distributed computation, but also many rich aggregation possibilities from previously computed sketches.
As a result, modern data pipelines rely extensively on the performance and functionality of such cardinality sketches. 

Increasingly, privacy concerns constrain the operation of data processing. 
Organizations demonstrating commitments to preserving users' privacy require that data collected from individuals be subject to appropriate mitigations before being passed to downstream processing. 
Specifically, protections such as differential privacy are used to protect sensitive data while still giving accurate query response. 

Although sketching techniques may appear to offer protection by reducing data, it is well-known that sketching alone does not automatically provide a privacy guarantee \citep{desfontaines2018cardinality}. 
The summaries---or even the estimates calculated from them---can leak considerable information about whether the specific items  belong to the underlying set.
Recently, it has been shown that the contents of sketches do meet a privacy standard \textit{if} the associated hash functions are not known to the observer \citep{choi2020differentially,smith2020flajolet,dickens2022nearly}.
However, it is not plausible to assume secret hash functions when the computation is shared among multiple entities in a large scale system. 
In particular, all participants must know the hash when working with sketches that will be merged, and using the same hash in multiple sketches generates correlated randomness that breaks the privacy guarantees.
This creates an important gap to make these high-throughput systems private. 
Previous attempts to construct privacy-preserving sketches \citep{pagh2020efficient}  do not offer practical mergeable sketches as the errors are too large (Section \ref{sec:evaluation}).

In this work, we present the Sketch-Flip-Merge (SFM) summaries, a practical, mergeable, and provably private approach to distinct-count sketching. 
In particular, we produce summaries that satisfy the strong definition of $\eps$--differential privacy (DP) \citep{dwork2006calibrating,dwork2008differential} even when the hash function is known publicly. 
By attaching the privacy guarantee to the summary itself---not just the cardinality estimate---we may safely release summaries corresponding to sensitive multisets, enabling safe cardinality estimation over any union of such sets using the privacy-preserving summaries in lieu of the original sensitive data.

The key to our approach is to adapt the sketch of \citet{flajolet1985probabilistic}, which is often referred to as either \textit{FM85} or \textit{probabilistic counting with stochastic averaging (PCSA)}. 
Although subsequent sketches such as HLL \citep{durand2003loglog, flajolet2007hyperloglog, heule2013hyperloglog}
further optimized the space usage,
squeezing the space makes them less amenable to privacy protection.
In contrast to PCSA where the simple, partitioned binary structure limits the sensitivity to a bit flip, these sketches store extremal hash values where
small changes to the input can cause big changes in the summary, requiring more noise and yielding less accurate results. Furthermore, our methods generalize to any bitmap based sketch.

\textbf{Related Work.} Privacy-preserving cardinality sketches have been the subject of several earlier works. 
While recent efforts provide DP guarantees for HLL-like sketches \citep{smith2020flajolet,dickens2022nearly}, they rely on random, secret hash functions that preclude the ability to merge sketches. 
Using a fixed, public hash, 
\citet{choi2020differentially} obtain a DP cardinality estimate from a LogLog sketch by adding noise to the cardinality estimator, but the sketch itself remains sensitive and unsafe for release or sharing. 
\citet{stanojevic2017distributed} design a DP algorithm for obtaining cardinality estimation on the union of two multisets using perturbed Bloom filters, but their method does not generalize and scale to the union of more than two multisets.

One line of work extends PCSA with randomized response and subsampling of items to achieve privacy \citep{tschorsch2013algorithm,nunez2019rrtxfm}. 
However, \citet{tschorsch2013algorithm} fails to achieve a DP guarantee, and \citet{nunez2019rrtxfm} does not address merging sketches.
\citet{kreuter2020privacy} design two sketches, including one based on PCSA. While their DP sketches cannot be merged to form a single sketch, multiple sketches may be used to estimate the union's cardinality if all sketches use the same privacy parameters.
The PCSA-based sketch of \citet[][Section~6]{pagh2020efficient} achieves DP and supports merging but is impractical. In our experiments, their estimator frequently 
failed to produce an estimate and returned impractically large errors.
Finally, \citet{desfontaines2018cardinality} 
give an impossibility result where both privacy and high accuracy are impossible, but only when many sketches are merged, which is consistent with our results.

\textbf{Contributions.} We propose two practical methods for constructing \textit{mergeable} DP cardinality sketches and obtaining cardinality estimates. The first uses a deterministic bit-merging operation used by \citet{pagh2020efficient}. We prove this merge requires a suboptimal form of randomized response, even after exponential improvement to the prior privacy analysis (Corollary~\ref{cor:equivalent-eps}).
Our main methodological contribution is a novel randomized merge allowing for up to a further 75\% variance reduction over the optimized deterministic merge. 
We generalize our randomized merge to perform to arbitrary bitwise operations on binary data that may be of independent interest. 
We also develop a composite likelihood-based estimator for cardinality and prove this estimator is asymptotically optimal for both private and non-private sketches based on PCSA.

\textbf{Outline.} We give a brief overview of PCSA sketching in Section~\ref{sec:background}, then define privacy and recap randomized response in Section~\ref{sec:privacy}. 
Merging sketches is enabled through the careful design of randomized response mechanisms and merge operations over collections of randomized bits in Section~\ref{sec:merging}. 
In Section~\ref{sec:sketches}, we propose a fast cardinality estimator for the private PCSA sketch and analyze its properties. 
We compare these methods with private and non-private alternatives in Section~\ref{sec:evaluation} and state conclusions in Section~\ref{sec:discussion}. All proofs are deferred to the appendices.

\textbf{Notation.} We write $[m] = \{1, \dots, m\}$. 
$\otimes$ denotes the Kronecker product. Logical operations are denoted $\lor$ (or), $\land$ (and), $\lxor$ (xor), and $\lnot$ (not). We use the natural logarithm $\log = \log_e$. Equality in distribution is denoted $\overset D=$.
The cardinality of a set $D$ is denoted $|D|$.

\section{Background and Problem Setup}
\label{sec:background}

Let  $D \in \X^N$ denote a multiset of $N$ items from some universe $\X$ of objects. The \textit{count-distinct problem} is the task of estimating the number of unique elements in $D$. 
That is, if $\fn{set}(D)$ denotes the support set of items in $D$, the count-distinct problem aims to approximate $n = |\fn{set}(D)|$ with a data sketch in bounded memory in a single pass over the data.
We consider the \emph{private} count-distinct problem for \emph{mergeable} sketches where the information in a sketch satisfies differential privacy (DP) and sketches can be merged to obtain a sketch of the union of underlying datasets.

We focus on solutions to the count-distinct problem in which sketches form a binary vector, subject to merge operations performed through element-wise logical operations (e.g., \textit{or}).
The class of sketches to which our methods apply include PCSA, linear counting \citep{whang1990linear}, Bloom filters \citep{broder2004network}, and Liquid Legions \citep{kreuter2020privacy}. These are particularly amenable to privacy enhancement through the application of randomized response \citep{warner1965randomized} but require careful design of merge operations for randomized bits. 
Although this excludes other commonly used sketches such as HyperLogLog and the $k$-minimum value sketch \citep{bar2002counting,giroire2009order}, the richer set of values stored in these sketches make them less suitable for privatization due to their high sensitivity and, hence, higher noise required for privacy.
In the remainder of this paper, we focus on the PCSA sketch of \citet{flajolet1985probabilistic}, noting that the results for constructing and merging private sketches in Sections~\ref{sec:privacy} and \ref{sec:merging} apply to related sketches through direct application or simple extensions.

The classical PCSA sketch takes the form of a matrix $S = \scal(D) \in \{0, 1\}^{B \times P}$ with $B$ buckets and precision parameter $P$. Given two independent, universal hash functions,
$h_1(x) \sim \uniform([B]), h_2(x) \sim \geometric(1/2)$, let $\fn{bucket}(x) = h_1(x), \fn{value}(x) = \min \{ P, h_2(x)\}$. Then each bit $S_{ij}$ is equal to 1 iff there exists $x \in D$ such that $\fn{bucket}(x) = i, \fn{value}(x) = j$. Some desirable properties of $S$ are immediate. First, $S$ relies only on the set of hashed values $\{ h_1(x) \}_{x \in D}$ and $\{ h_2(x) \}_{x \in D}$. 
Hence, it is invariant both to repetitions in $D$ and to the order in which the elements of $D$ are processed. 
Additionally, two sketches $\scal(D_1)$ and $\scal(D_2)$ may be merged via a simple bitwise-\textit{or}, $\scal(D_1) \lor \scal(D_2) = \scal(D_1 \cup D_2)$, as each entry $S_{ij}$ in the merged sketch is equal to $1$ iff there is an item $x$ in at least one of $D_1, D_2$ for which $\fn{bucket}(x) = i, \fn{value}(x) = j$.

Importantly, when an adversary knows $h_1$ and $h_2$, the sketch $\scal(D)$ reveals information about elements in $D$. 
For example, any $x \in \X$ for which $S_{\fn{bucket}(x), \fn{value}(x)} = 0$ cannot belong to $D$. 
In what follows, we extend the PCSA sketch to minimize this sort of privacy leakage.

\section{Private Sketches}
\label{sec:privacy}

Differential privacy (DP) \citep{dwork2006calibrating,dwork2008differential} offers a strong and quantifiable notion of privacy. DP mandates that algorithms (\textit{privacy mechanisms}) acting on a dataset $D$ must be randomized---typically through the addition of some carefully tuned noise---so that the distribution of a privacy mechanism's output cannot be significantly influenced by a single input record. As a result, the ability to reverse-engineer information about a single record is limited, and any analysis performed using only the output of the algorithm also satisfies DP. The strength of the DP guarantee is quantified by the parameter $\eps > 0$, often called the privacy budget, with smaller $\eps$ offering stronger privacy. 
\begin{definition}[\citet{dwork2006calibrating}]
\label{def:dp}
A randomized algorithm $\M$ is said to satisfy \textbf{$\eps$--differential privacy (DP)} if for any two neighboring databases $D, D'$ and any set of outputs $E \subseteq \mathrm{Range}(\M)$, we have:
\[
\Pr(\M(D) \in E) \leq e^\eps \, \Pr(\M(D') \in E) .
\]
\end{definition}
In the count-distinct problem, we say two multisets $D, D'$ are neighbors if $D'$ can be obtained by adding or removing one unique item to $D$.
Given two neighboring multisets $D, D'$ and their corresponding PCSA sketches $\scal(D), \scal(D')$, it follows from the definition of PCSA that these sketches must agree on all but at most one bit. 
To create DP sketches from $\scal(D)$, then, we consider general DP mechanisms applied to vectors of $\{0, 1\}$ bits, where two vectors $x, x' \in \{0, 1\}^d$ neighbor if they differ on at most one bit (i.e., have \textit{sensitivity} 1).

When restricting our attention to mechanisms whose input and output are both a single bit, every DP mechanism can be viewed as an instance of randomized response (RR) \cite{warner1965randomized}.
We describe a generalized form of RR as follows. 
Let $\F_{p,q}$ denote a general bit-flipping algorithm, parameterized by two probabilities $p$ and $q$, and where classical RR (in the style of Warner) is recovered when $q=1-p$:
\[
\F_{p,q}(x) \sim \begin{cases}
    \bernoulli(p), & x = 1 \\
    \bernoulli(q), & x = 0
\end{cases} .
\]

\begin{theorem}
\label{thm:dp-bit-flip}
Assume $q \leq 1/2 \leq p$. Applied to vectors with sensitivity 1, the algorithm $\M_{p, q} : \{0, 1\}^d \to \{0, 1\}^d$ that independently applies $\F_{p, q}$ to each element of its input is $\eps$-DP if and only if:
\begin{equation}
\label{eq:dp-pq-constraints}
p, q \in (0, 1) \quad \text{and} \quad \max \left\{ \frac{p}{q}, \frac{1-q}{1-p} \right\} \leq e^\eps .
\end{equation}
\end{theorem}

Theorem~\ref{thm:dp-bit-flip} provides an entire family of privacy mechanisms, any $\M_{p,q}$ with  $p, q$ satisfying Eq. \ref{eq:dp-pq-constraints}, that can be applied to a PCSA sketch or any bit vector to make its output $\eps$-DP.
Our contribution is then to address several important questions: How can we merge two sketches if their bits have been perturbed via $\F_{p, q}$? How can we estimate cardinality from noisy sketches? And how should we choose $p$ and $q$?

\section{Merging Perturbed Bit Vectors}
\label{sec:merging}

While a randomized response mechanism $\M_{p,q}$ converts a PCSA sketch to a private equivalent, this breaks PCSA's merge operation.
For ordinary PCSA and multisets $D_1, D_2$, the bitwise-\textit{or} 
$\scal(D_1) \lor \scal(D_2) = \scal(D_1 \cup D_2)$
defines a merge operation on sketches that yields the same sketch that would be obtained by first taking the union. 
However, the same operation on noisy sketches does not satisfy this desirable property. 
We develop merge operations on noisy sketches and identify under what conditions they exist. 
In particular, Theorem \ref{thm:unique-deterministic-merge} shows that if a merge operation is deterministic, then \textit{xor} $(\lxor)$ is the only possible merge on noisy sketches, and it only works for certain choices of the mechanism $\M_{p, q}$. 
We show these choices imply that, at best, such a noisy sketch's cardinality estimates have $4\times$ worse variance than that for regular PCSA on the same sized sketch, even if the privacy budget is near-infinite. 
Our main contribution is to provide a novel \textit{randomized} merge operation that adds less noise to the sketch. 
Furthermore, we generalize this operation to perform arbitrary boolean operations on noisy bit vectors.

\subsection{Deterministic Merging}
\label{sec:merging-deterministic}
Applying the standard randomized response mechanism to a PCSA sketch breaks mergeability. PCSA merges sketches using bitwise-\textit{or}, and in the presence of RR noise, the \textit{or} operation results in non-RR noise that biases bits towards 1. 
\citet{pagh2020efficient} address this by replacing
bitwise-\textit{or} ($\lor$) with bitwise-\textit{xor} ($\lxor$)
 operations whenever the sketch is updated or merged. 
However, the \textit{xor} operation destroys cardinality information. 
In particular, the \textit{xor} of a PCSA sketch with itself is the empty sketch. 
Rather than ensuring sketches are invariant to duplicates, they ensure the \emph{distribution} of (merged) sketches are invariant. They do this by subsampling items (including duplicates) independently with probability $1/2$. This effectively encodes bits that were 1 in PCSA as $\bernoulli(1/2)$ values, while 0 bits remain 0.
Unfortunately, this subsampling operation introduces a lot of noise. 
Figure~\ref{fig:head-to-head} shows that even for large $\eps$ 
the resulting cardinality estimates have $4$ times the variance.

We show that this penalty on the accuracy is inherent for any \emph{deterministic} merge. 
Theorem~\ref{thm:unique-deterministic-merge} shows \textit{xor} is, in fact, the only possible way to merge deterministically, so that randomized merges are the only way to improve merging. Our analysis also improves the \citet{pagh2020efficient} sketch by significantly reducing the noise required for an $\eps$-DP privacy guarantee and demonstrates how to merge sketches with different privacy budgets.

\begin{theorem}
\label{thm:unique-deterministic-merge}
Let $f_1 = \F_{p_1,q_1}, f_2 = \F_{p_2,q_2}, f_3 = \F_{p_3,q_3}$, and let $\circ : \{0, 1\}^2 \to \{0, 1\}$ denote a deterministic and symmetric operation. 
The following conditions may only be satisfied simultaneously if $\circ = \lxor$ and $p_1 = p_2 = 1/2$:
\begin{enumerate}
    \item $f_1, f_2$ are $\eps_1$-DP and $\eps_2$-DP for $\eps_1, \eps_2 < \infty$.
    \item $f_1(x) \circ f_2(y) \overset D= f_3(x \lor y)$.
    \item $f_i(0) \overset D\neq f_i(1)$ for $i = 1, 2, 3$.
\end{enumerate}
\end{theorem}
Using our general family of RR mechanisms, 
we define a mechanism that adds noise to a PCSA sketch to provide privacy (Lemma~\ref{thm:xor-rr-dp})
while preserving mergeability (Theorem~\ref{thm:xor-rr-merge}). Corollary~\ref{cor:equivalent-eps} shows our privacy analysis is much tighter than that of \citet{pagh2020efficient}. 

\begin{definition}
\label{def:xor-rr}
For $\eps > 0$, let $\Mxor_\eps : \{0, 1\}^d \to \{0, 1\}^d$ denote the mechanism that independently applies an asymmetric random response $\F_{p, q}$ to each element of its input with $p = 1/2, q = 1/(2e^\eps)$.
\end{definition}

\begin{lemma}
\label{thm:xor-rr-dp}
$\Mxor_\eps$ is $\eps$--differentially private.
\end{lemma}

\begin{theorem}
\label{thm:xor-rr-merge}
$\Mxor_{\eps_1}(x) \lxor \Mxor_{\eps_2}(y) \overset D= \Mxor_{\eps^*}(x \lor y)$, where
$\eps^* = -\log(e^{-\eps_1} + e^{-\eps_2} - e^{-(\eps_1 + \eps_2)})$.
\end{theorem}

\begin{corollary}
\label{cor:equivalent-eps}
Let $\Mps_\eps$ denote the $\eps$-DP privacy mechanism of \citet[Section~6]{pagh2020efficient}. Then $\Mxor_\eps = \Mps_{2(\exp(\eps) - 1)}$.
\end{corollary}

This tighter privacy analysis%
\footnote{It is proven in the appendix of \citet{pagh2020efficient} that $q = 1 / (e^\eps + 1)$ satisfies $\eps$-DP, although the recommendation and main results in the paper rely on the choice of $q = 1 / (2 + \eps)$. Our recommendation of $q = 1 / (2e^\eps)$ is optimal under DP constraints.}
dramatically reduces noise added to achieve the privacy guarantee, effectively increasing the privacy budget by at least a factor of $2$. 
Pragmatically, Figure~\ref{fig:head-to-head} shows that error increases exponentially as $\eps \to 0$.

\subsection{Randomized Merging}
\label{sec:merging-randomized}
Theorem~\ref{thm:unique-deterministic-merge} showed that a deterministic merge is only possible if the 1-bits in a PCSA sketch are randomized to $\bernoulli(1/2)$ values. Thus, even if the privacy budget is nearly infinite, the mergeable DP sketch must add significant noise to the base PCSA sketch.
We show that by moving randomness from the base sketch to the merge procedure, we can achieve lower overall noise while using the standard randomized response mechanism (Definition~\ref{def:sym-rr}).

\begin{definition}
\label{def:sym-rr}
For $\eps > 0$, we denote by $\Msym_\eps : \{0, 1\}^d \to \{0, 1\}^d$ the mechanism that independently applies the standard RR mechanism $\F_{p, 1-p}$ to each element of its input with $p = e^\eps/(e^\eps + 1)$.
\end{definition}

\begin{lemma}
\label{thm:sym-rr-dp}
$\Msym_\eps$ is $\eps$--differentially private.
\end{lemma}

A merge is a randomized algorithm $g_{\eps_1, \eps_2} : \{ 0, 1 \}^2 \to \{0, 1\}$ 
that commutes with $\lor$ in the following sense:
\[
g_{\eps_1,\eps_2}(\Msym_{\eps_1}(x), \Msym_{\eps_2}(y)) \overset D= \Msym_{\eps^*}(x \lor y).
\]
Since $g_{\eps_1, \eps_2} $ is a random mapping from pairs of bits to single bits, we can represent it as a $4 \times 2$ Markov transition matrix. A valid merge operation is the solution of the resulting matrix equality, with $\eps^*$ a free parameter.
We obtain a optimal randomized merge operation for $\Msym$ by solving for the largest $\eps^*$ that generates a valid solution, which is given by:

\begin{theorem}
\label{thm:sym-rr-merge}
Assume $\eps_1, \eps_2 > 0$. Let $q(\eps) = (e^\eps + 1)^{-1}$,
\[
\eps^* = -\log(e^{-\eps_1} + e^{-\eps_2} - e^{-(\eps_1 + \eps_2)}), \ q^* = q(\eps^*),
\]
\[
K_i = \begin{bmatrix} 1 - q(\eps_i) & q(\eps_i) \\ q(\eps_i) & 1 - q(\eps_i) \end{bmatrix} \text{ for }  i \in \{1, 2\} \text{, and}
\]
\[
v^* = (q^*, 1-q^*, 1-q^*, 1-q^*)^T ,
\]
Letting $\otimes$ denote the Kronecker product, define:
\[
(t_{00}, t_{01}, t_{10}, t_{11})^T = ( K_1^{-1} \otimes K_2^{-1} ) \, v^* , \text{ and}
\]
\[
g_{\eps_1,\eps_2}(a, b) \sim \bernoulli(t_{ab}), \quad a, b \in \{0, 1\}.
\]
Then $g_{\eps_1,\eps_2}(\Msym_{\eps_1}(x), \Msym_{\eps_2}(y)) \overset D= \Msym_{\eps^*}(x \lor y)$, where $g$ is taken bitwise and independently.
\end{theorem}

When the original vectors $x$ and $y$ are visible in addition to the merged vector, the $\eps^*$ parameter of Theorems~\ref{thm:xor-rr-merge} and \ref{thm:sym-rr-merge} is best interpreted as a measure of utility in the merged sketch, rather than a privacy budget, since by the post-processing invariance of DP \citep{dwork2006calibrating}, no additional privacy leakage occurs from the release of the merged vector. It is for this reason we seek the maximal $\eps^*$ in merging. Noting that Theorems~\ref{thm:xor-rr-merge} and \ref{thm:sym-rr-merge} produce identical $\eps^*$ and that $\Msym_\eps$ is less noisy than $\Mxor_\eps$, $\Msym$ remains the preferred mechanism after merging.

\begin{remark}
\label{remark:k-way-merge}
By induction, the merges prescribed in Theorems~\ref{thm:xor-rr-merge} and \ref{thm:sym-rr-merge} allow any $k$ bit vectors of equal length $x_1, \dots, x_k$ privatized using $\eps_1, \dots, \eps_k$ to be merged, resulting in a vector  $v = (x_1 \lor \dots \lor x_k)$ privatized with
\[
\eps^* = -\log \left( 1 - \prod_{i=1}^k(1 - e^{-\eps_i} ) \right) .
\]
A natural question is whether there exists a randomized merge algorithm $g_{\eps_1, \dots, \eps_k} : \{0, 1\}^k \to \{0, 1\}$ that satisfies a property like Theorem~\ref{thm:sym-rr-merge} with a larger $\eps^*$ than given by induction over the pairwise merges. In Appendix~\ref{sec:proofs}, we prove a more general form of Theorem~\ref{thm:sym-rr-merge} (Theorem~\ref{thm:sym-merge-general}), which answers this question in the negative.
\end{remark}

\subsection{General Boolean Operations}
\label{sec:general-boolean-ops}
We briefly switch focus from distinct counting to present a generalization to Boolean operations under randomized response that may be of foundational interest, e.g., in designing intersection operations. 
In distinct-count sketches, set unions correspond to bitwise-\textit{or} operations, and the challenge posed by privacy is performing an equivalent operation over noisy bits.
PCSA, like other mergeable sketches, defines a homomorphism $\scal$ from multisets to sketches. 
The commutative diagram below illustrates this mergeability property, since 
it does not matter which path one takes from $D_1,D_2$ to $v_1 \lor v_2$.
Likewise, our merge operation $g$  from Theorem~\ref{thm:sym-rr-merge} ensures that the privacy mechanism $\M$  makes the diagram commute.
By preserving the structure of the union operation, inferences about the cardinality of the union can be made from merged, private sketches. 

\begin{center}
    \begin{tikzpicture}
    \matrix(m)
    [
            matrix of math nodes,
            row sep    = 2.5em,
            column sep = 2.5em
    ]
    {
          \text{\small{\textsc{Sets}}} & \text{\small{\textsc{Bit Vectors}}} & \text{\small{\textsc{DP Bit Vectors}}} \\[-2.25em]
          D_1,D_2 & v_1, v_2 & \M_{\eps}(v_1), \M_{\eps}(v_2)\\
          D_1 \cup D_2  & v_1 \lor v_2 &  \M_{\eps^*}(v_1 \lor v_2) \\
    };
        \path[-stealth]
         (m-2-1) edge node [left] {$\cup$} (m-3-1)
        (m-2-1.east |- m-2-2)
          edge node [above] {$\scal$} (m-2-2)
        (m-3-1) edge node [above] {$\scal$}
          (m-3-2)
         (m-2-2) edge node [left] {$\lor$} (m-3-2)
         (m-2-2) edge node [above] {$\M_\eps$} (m-2-3)
         (m-3-2) edge node [above] {$\M_{\eps^*}$} (m-3-3)
         (m-2-3) edge node [left] {$g$} (m-3-3)
         ;
    \end{tikzpicture}
\end{center}

Here, we generalize the \textit{or} ($\lor$) merge under $\Msym$ to any logical operation. In particular, Corollary~\ref{thm:and-sym-rr-merge} shows a simple change in our target probabilities $v^*$ yields the appropriate randomized merge for \textit{and} ($\land$), while Lemmas~\ref{thm:xor-msym} and \ref{thm:not-msym} in Appendix~\ref{sec:and-merging-appendix} demonstrate a merge for \textit{xor} ($\lxor$) and show that \textit{not} ($\lnot$) commutes with $\Msym$.

\begin{corollary}
\label{thm:and-sym-rr-merge}
Assume the setting of Theorem~\ref{thm:sym-rr-merge}, but set $v^* = (q^*, q^*, q^*, 1-q^*)^T$. Then $g_{\eps_1,\eps_2}(\Msym_{\eps_1}(x), \Msym_{\eps_2}(y)) \overset D= \Msym_{\eps^*}(x \land y)$, where $g$ is taken bitwise and independently.
\end{corollary}

Formally, for each binary logical operator $\square$, there is a function $\eps^*_\square(\eps_1, \eps_2)$ combining two privacy budgets that endows pairs of bit vectors and privacy budgets with the semigroup structure
$(v_1, \eps_1) \cdot (v_2, \eps_2) \defn (v_1 \,\square\, v_2, \eps^*_\square(\eps_1, \eps_2))$.
The privacy mechanism $\Msym$ then defines a mapping
$\phi(v, \eps) := (\Msym_\eps(v), \eps)$ that is a homomorphism from this semigroup to its noisy counterpart.
These operations are summarized in Table~\ref{tab:boolean-operations-on-msym}.

\begin{table}
    \centering
    \caption{Boolean Operations on Bit Vectors under $\Msym$}
    \vskip 0.15in
    \begin{small}
    \begin{sc}
    \begin{tabular}{ccc}
        \hline
        Op. ($\square$) & DP Op. & $\eps^*_\square(\eps_1, \eps_2)$ \\
        \hline
        $\lnot$ & Lem.~\ref{thm:not-msym} & --- \\
        $\lor$ & Thm.~\ref{thm:sym-rr-merge} & $-\log(e^{-\eps_1} + e^{-\eps_2} - e^{-(\eps_1 + \eps_2)})$ \\
        $\land$ & Cor.~\ref{thm:and-sym-rr-merge} & $-\log(e^{-\eps_1} + e^{-\eps_2} - e^{-(\eps_1 + \eps_2)})$ \\
        $\lxor$ & Lem.~\ref{thm:xor-msym} & $\log \left( \frac{1 + e^{\eps_1 + \eps_2}}{e^{\eps_1} + e^{\eps_2}} \right)$ \\
        \hline
    \end{tabular}
    \end{sc}
    \end{small}
    \vskip -0.1in
    \label{tab:boolean-operations-on-msym}
\end{table}

Our randomized merging technique is crucial in supporting general Boolean operations. Unlike our randomized merge operations, Corollary~\ref{thm:impossible-deterministic-and-or} shows no \textit{deterministic} operations $\circ$ and $\bullet$ can define merge operations for both \textit{or} ($\lor$) and \textit{and} ($\land$) under the same RR mechanism.

\begin{corollary}
\label{thm:impossible-deterministic-and-or}
Assume the setting of Theorem~\ref{thm:unique-deterministic-merge}. Let $\bullet : \{0, 1\}^2 \to \{0, 1\}$ denote a deterministic and symmetric operation. It is impossible to satisfy conditions (1)--(3) of Theorem~\ref{thm:unique-deterministic-merge} in addition to the following:
\begin{enumerate}\setcounter{enumi}{3}
\item $f_1(x) \bullet f_2(y) \overset D= f_3(x \land y)$.
\end{enumerate}
\end{corollary}

\section{Cardinality Estimation}
\label{sec:sketches}

The Sketch-Flip-Merge method developed so far satisfies privacy and mergeability requirements, but it remains to show how SFM summaries may be used to estimate cardinalities.
We develop a composite likelihood--based estimator that is consistent and asymptotically optimal. 
We give an analytic estimator of the error that closely matches the true error in our experiments.

Likelihood-based approaches to cardinality estimation have been used in the non-private setting \citep{clifford2012statistical,lang2017back,ertl2017new,ting2019approximate}, where they have demonstrated greater accuracy than competing estimators for PCSA. 
While true maximum likelihood estimation of $n$ given a sketch is
computationally infeasible due to non-independence of bits in the sketch
\citep{ting2019approximate,ertl2017new},
the marginal likelihood for any bit is easy to derive.
Similar to \citet{ting2019approximate}, we derive a composite marginal likelihood estimator \citep{lindsay1988composite,varin2011overview} for $n$.

Let $C_{ij}$ denote the number of unique items $x \in D$ 
mapped to bucket $i$ and value $j$ in the sketch
and $\rho_{ij}$ be the probability an item is mapped to that location. Then, assuming the use of universal random hashes, the following generative process describes the SFM summary $T = \M_{p, q}(\scal(D))$.
\begin{align*}
(C_{11}, \dots, C_{BP}) &\sim \multinomial(n, \rho_{11}, \ldots, \rho_{BP}) \\
T_{ij} \mid C_{ij} &\sim \F_{p,q}( 1(C_{ij} > 0) )
\end{align*}

While the joint distribution of $ \{ C_{ij} \}_{i,j}$, and hence the observed $ \{ T_{ij} \}_{i,j}$,
involves an intractable sum over integer partitions of $n$ into at most $BP$ parts, the marginal distribution of a single bit $T_{ij}$ is easy to compute. Note that cell $C_{ij}$'s probability of occupancy $\rho_{ij}$  depends only on its level $j$, with
$\rho_{ij} = 2^{-\min\{j, P-1\}} / B$.
Then $\Pr(C_{ij}=0 \mid n) = \gamma_j^n$, where $\gamma_j = 1-\rho_{ij}$, and
\[
T_{ij} \sim \bernoulli\left( p (1-\gamma_j^n) + q \gamma_j^n \right).
\]
The composite marginal log-likelihood \citep{lindsay1988composite,varin2011overview} replaces the true log-likelihood by the surrogate $\ell_{p, q}(n; t)$ that sums over marginal log-probabilities,
\[
\ell_{p, q}(n; t) = \sum_{ij} (1 - t_{ij}) \log \left( 1 - p + (p-q) \gamma_j^n \right) + \sum_{ij} t_{ij} \log \left( p - (p-q) \gamma_j^n \right),
\]
where $t = \M_{p, q}(\scal(D))$ denotes a realized SFM summary.
The corresponding composite maximum likelihood estimator is
$\hat n = \max_n \ell_{p,q}(n; t)$
and can be optimized by Newton's method.
The required first and second
derivatives of $\ell_{p,q}$ are
\begin{align*}
\ell'_{p, q}(n; t) &= \sum_{ij} (1 - t_{ij}) \tfrac{(p-q) \gamma_j^n \log(\gamma_j)}{1 - p + (p-q)\gamma_j^n} - \sum_{ij} t_{ij} \tfrac{(p-q) \gamma_j^n \log(\gamma_j)}{p - (p-q)\gamma_j^n} , \\
\ell''_{p, q}(n; t) &= \sum_{ij} (1 - t_{ij}) \tfrac{(1-p) (p-q) (\log \gamma_j)^2 \gamma_j^n}{(1 - p + (p-q) \gamma_j^n)^2} - \sum_{ij} t_{ij} \tfrac{p (p-q) (\log \gamma_j)^2 \gamma_j^n}{(p - (p-q) \gamma_j^n)^2} .
\end{align*}

In the absence of privacy (i.e., $p = 1, q = 0$), $\ell_{p, q}(\, \cdot \, ; t)$ is strictly concave. 
While this is not true in the private case, Theorem~\ref{thm:shape-of-loglik} states that the expectation of $\ell_{p, q}$ remains concave over the interval $(0, n+\Delta)$ for some $\Delta > 0$.

\begin{theorem}
\label{thm:shape-of-loglik}
Let $D$ be a multiset such that $|\fn{set}(D)| = n$. Let $T = \M_{p, q}(\scal(D))$. Let $f(\hat n) = \E[\ell_{p, q}(\hat n; T)]$, where the expectation is taken over the randomness of the hash functions $h_1, h_2$ and the privacy mechanism $\M_{p, q}$. Then $f(\hat n)$ attains a global maximum at $\hat n = n$ and is concave on an interval containing $(0, n]$ in its interior.
\end{theorem}

\subsection{Theoretical Results}
\label{sec:theory}
We choose to use composite marginal likelihood due to its attractive theoretical properties. In particular, the use of a true likelihood, even if incomplete, ensures that cardinality estimates are asymptotically consistent, and the Hessian of the composite likelihood provides an estimate of the variance \citep{ting2019approximate}.
We further show the cardinality estimates are asymptotically optimal in the typical case when the cardinality is large relative to the sketch size.

\begin{theorem}
\label{thm:sketch convergence in probability}
Let $S_n$ denote a PCSA summary of $n$ distinct items
with $B(n)$ buckets and $P=\infty$ levels.
Let $\overline{S}_n$ denote a modified PCSA summary 
of $\poisson(n)$ distinct items,
and $\tilde{S}_n$ denote one where the composite marginal likelihood is the true likelihood. 
If $B \log B =o(\sqrt{n})$, then there exists modified PCSA summaries $\tilde{S}_n, \overline{S}_n$
where
\[
\Pr(S_n = \tilde{S}_n = \overline{S}_n) \to 1 \quad \mbox{as $n \to \infty$.}
\]
\end{theorem}

\begin{corollary}
\label{cor:asymptotic efficiency}
The composite likelihood estimator of the SFM sketch's cardinality is asymptotically efficient in the asymptotic regime in Theorem \ref{thm:sketch convergence in probability}.
\end{corollary}

We outline the proofs and provide details in Section~\ref{sec:sketch-proofs} of Appendix~\ref{sec:proofs}. The main difficulty is that entries in an SFM summary are dependent, since each item can only be allocated to one cell.
By constructing a coupling between sketches $S_n, \overline{S}_n$ with dependent and independent entries, we show they are asymptotically equal. In these coupled processes, the bucket with maximum difference in item allocations can only differ by only a small amount, $O_p(\sqrt{n} /B \log B)$. By showing an item updates its bucket's sketch values with probability $O(1/v_i)$
(where $v_i$ is the number of items in bucket $i$),
we conclude the coupled sketches are, in fact, equal with probability going to 1
when the average bucket allocation $v_i \approx n/B$ grows fast enough to make the small differences in item allocation irrelevant.
Since $\overline{S}_n, \tilde{S}_n$ have independent bits, we couple them via the inverse CDF method
and directly bound the probability that they differ.
Since the sketches are the same asymptotically, applying the exact same RR noise to them implies their privatized versions are the same, and any estimator on them has the same asymptotic sampling distribution.
Therefore, the asymptotically efficient MLE for the independent-entry sketch, i.e., the composite likelihood estimator, is also asymptotically efficient for the true SFM summary $S_n$.

\begin{remark}
\label{remark:poissonization}
This result also proves the MLE derived under the approximation that each bin has $\poisson(n/B)$ items is asymptotically efficient. This can be extended to HyperLogLog and other sketches to show pseudo-likelihood based estimators \citep{ertl2017new, ting2019approximate} that have good empirical properties are, in fact, asymptotically optimal.
\end{remark}

\textbf{Error Estimation.} Like the Fisher information matrix for MLE's, the inverse Godambe (or sandwich) information provides a consistent estimate of the estimator's variance. The Godambe information is $G(n) = H(n) J^{-1}(n) H(n)$
where $-J(n)$ is the Hessian of the expected log composite likelihood at $n$ and 
$H(n) = \Var(\ell_{p,q}'(n, T))$ is the variance of the composite score functions.
In the non-private case, \citet{ting2019approximate} demonstrated that composite marginal likelihood variance estimates for HyperLogLog based on Fisher information and Godambe information are nearly identical for large cardinalities and that the Fisher information overestimates the variance at small cardinalities due to the negative dependence of buckets. Figure \ref{fig:error-by-n} shows this overestimation is much less pronounced when independent randomized response noise is added. Thus, we use only the Hessian to define the \textit{estimated standard error} as
\begin{align}
\label{eq:fisher-error}
\estse_{p, q}(B, P, n) &= \sqrt { 1 / \E[ -\ell''_{p, q}(n; T) ] } \\
    &= \left[ B (p-q) \sum_{j=1}^P (\log \gamma_j)^2 \gamma_j^n \left( \tfrac{p}{p - (p-q) \gamma_j^n} - \tfrac{1-p}{1 - p + (p-q) \gamma_j^n} \right) \right]^{-1/2} , \nonumber
\end{align}
where $T = \M_{p, q}(S)$ for a random PCSA sketch $S$ of size $B \times P$ and cardinality $n$.
Figure \ref{fig:head-to-head} in Section \ref{sec:evaluation} demonstrates empirically that our error estimates $\estse_{p, q}(B, P, n)$ are a good approximation for the error.

\section{Evaluation}
\label{sec:evaluation}

We evaluate our methods on both real-world and synthetic datasets. 
We demonstrate empirically that the SFM summary is the first mergeable $\eps$-DP distinct counting sketch with practical performance, since the errors for the \citet{pagh2020efficient} sketch are impractically large. 
Among private sketches, our novel randomized-merge sketch construction dominates the deterministic-merge sketches. Thus, 
our improvements on the construction, estimation, and privacy analysis yield practical gains.
Moreover, our theoretical error closely approximates empirical error.

\subsection{Experiment Setup}
We consider four different private distinct counting sketches in our experiments. Among our methods, \ussym{} pairs $\Msym_\eps$ with our randomized merge procedure, while \usxor{} pairs $\Mxor_\eps$ with the deterministic \textit{xor} merge. Both SFM methods use the estimator of Section~\ref{sec:sketches}. We compare these methods against the sketch and estimator of \citet{pagh2020efficient} implemented two ways: \psloose{} constructs sketches using $\Mps_\eps = \M_{p, q}$ with $p = 1/2, q = 1 / (2 + \eps)$ as prescribed in \citet{pagh2020efficient}, while \pstight{} uses the tightened $\Mxor_\eps$ (Definition~\ref{def:xor-rr}). By Corollary~\ref{cor:equivalent-eps}, the sketch construction of \pstight{} at $\eps$ is equivalent to that of \psloose{} at $\eps' = 2(e^\eps - 1)$.

We also consider two non-private sketches as baseline comparisons in our final experiment, noting that we should not expect the accuracy of a private sketch to be as strong as a non-private one. 
In what follows, \textit{FM85} denotes non-private PCSA using the estimator of \citet{flajolet1985probabilistic}, and \textit{HLL} denotes HyperLogLog \citep{flajolet2007hyperloglog}. We compare HLL sketches against PCSA sketches at equal bucket counts, noting that the HLL sketches are smaller per bucket than the corresponding PCSA sketches.

We measure estimation error primarily in the form of \textit{relative root mean squared error (RRMSE)}, defined as
\[
\mathrm{RRMSE}(\hat n_1, \dots, \hat n_m; n) = \frac 1n \sqrt{\frac 1m \sum_{i=1}^m(\hat n_i - n)^2} .
\]
We also measure the \textit{relative efficiency} of two methods as the ratio of their mean squared error,
\[
\mathrm{RE}(\hat{n}^{(1)}, \hat{n}^{(2)}) = \frac{\mathrm{MSE}(\hat{n}^{(2)})}{
\mathrm{MSE}(\hat{n}^{(1)})} .
\]
If two sketches have unbiased estimators, a relative efficiency of $r$ indicates that the less efficient sketch must asymptotically use $r$ times more buckets to get the same accuracy. 
This is because the asymptotic MSE (variance) decreases proportionally to $1/B$ for these estimators.

The simulations use sketches with dimensions $B=4096$, $P=24$ by default, using the xxHash64 \citep{xxhash} hash function, averaged over $m = 1000$ trials.

\textbf{Modification to \citet{pagh2020efficient}.} In our experiments, their original estimator frequently failed to produce an estimate. For a desired error tolerance $\beta$, the method computes intervals for all $P$ levels of the sketch and intersects them to produce an estimate. 
This intersection was frequently empty for small $\beta$. To patch this, we search for the smallest $\beta$ resulting in a non-trivial intersection.
We use the midpoint of this interval to estimate $n$.

\textbf{Data Sources.} 
Our experiments use both synthetic and real data. 
Synthetic data consist of random sets of integers with a fixed cardinality. 
Real data is taken from the BitcoinHeist paper \citep{akcora2019bitcoinheist}, which provides a database
of Bitcoin transactions to $n =$ 2,631,095 unique addresses. 
The dataset is available in the UCI repository under the CC BY 4.0 license.
We note that distinct-count sketches are insensitive to the value distribution of  inputs since  values are hashed as part of the processing.

\subsection{Results}

\begin{figure}
    \centering
    \includegraphics[width=.5\textwidth]{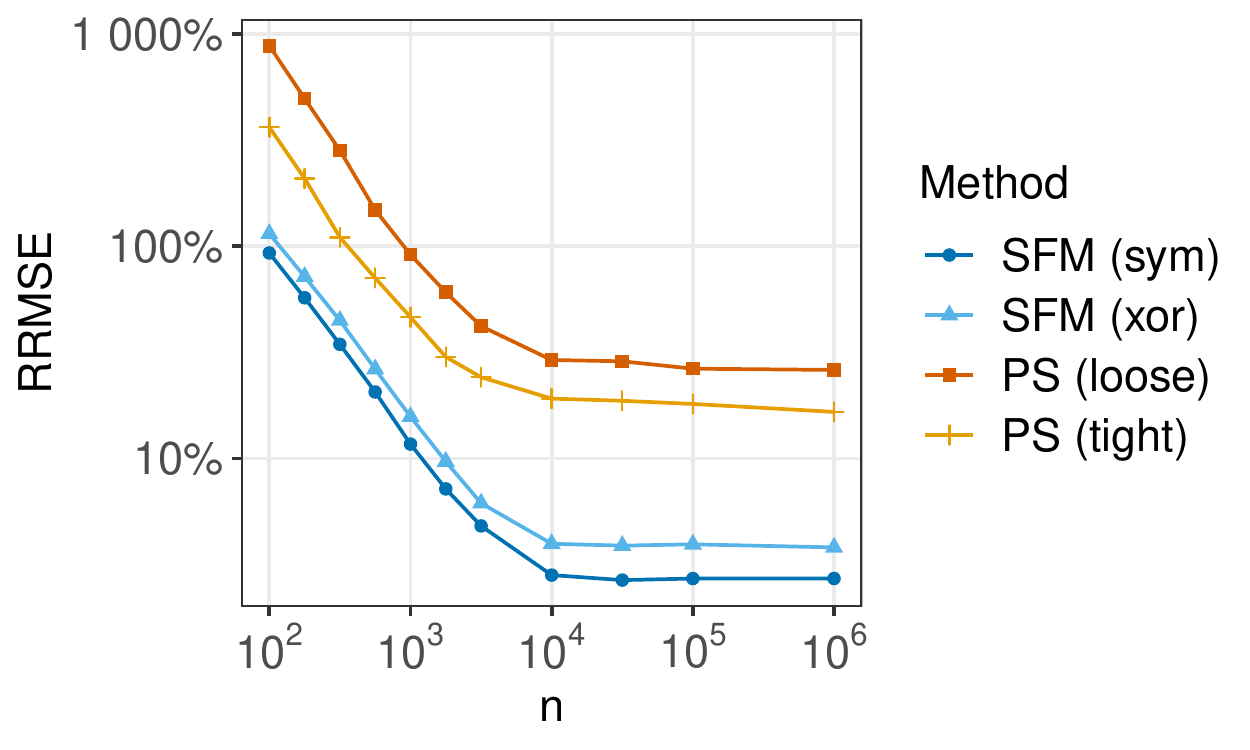}
    \caption{RRMSE vs. $n$ at $\eps = 1$ on log--log axes, compared across the four methods. RRMSE stabilizes for large $n$.}
    \label{fig:error-by-n-all}
\end{figure}

\begin{figure}
    \centering
    \includegraphics[width=.5\textwidth]{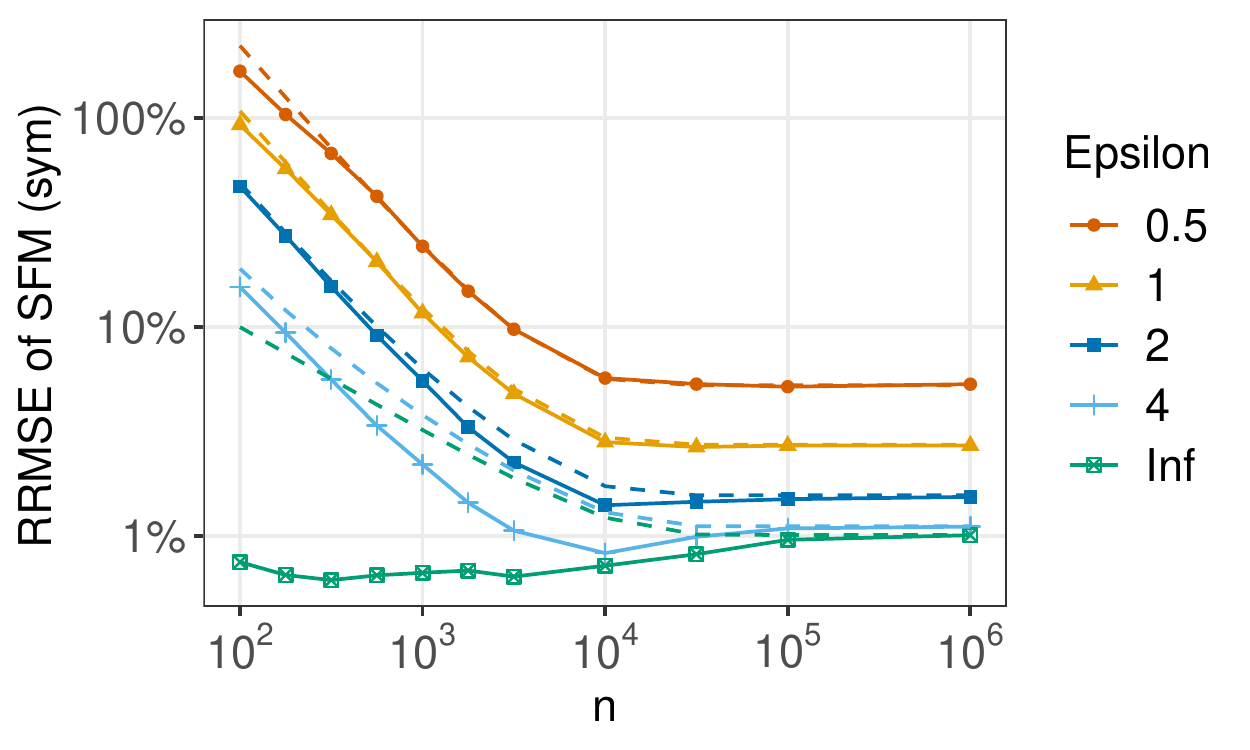}
    \caption{Dashed lines show estimated relative error, $\estse/n$, which is highly accurate for large $n$ and small $\eps$.}
    \label{fig:error-by-n}
\end{figure}

\begin{figure}
    \centering
    \includegraphics[width=.5\textwidth]{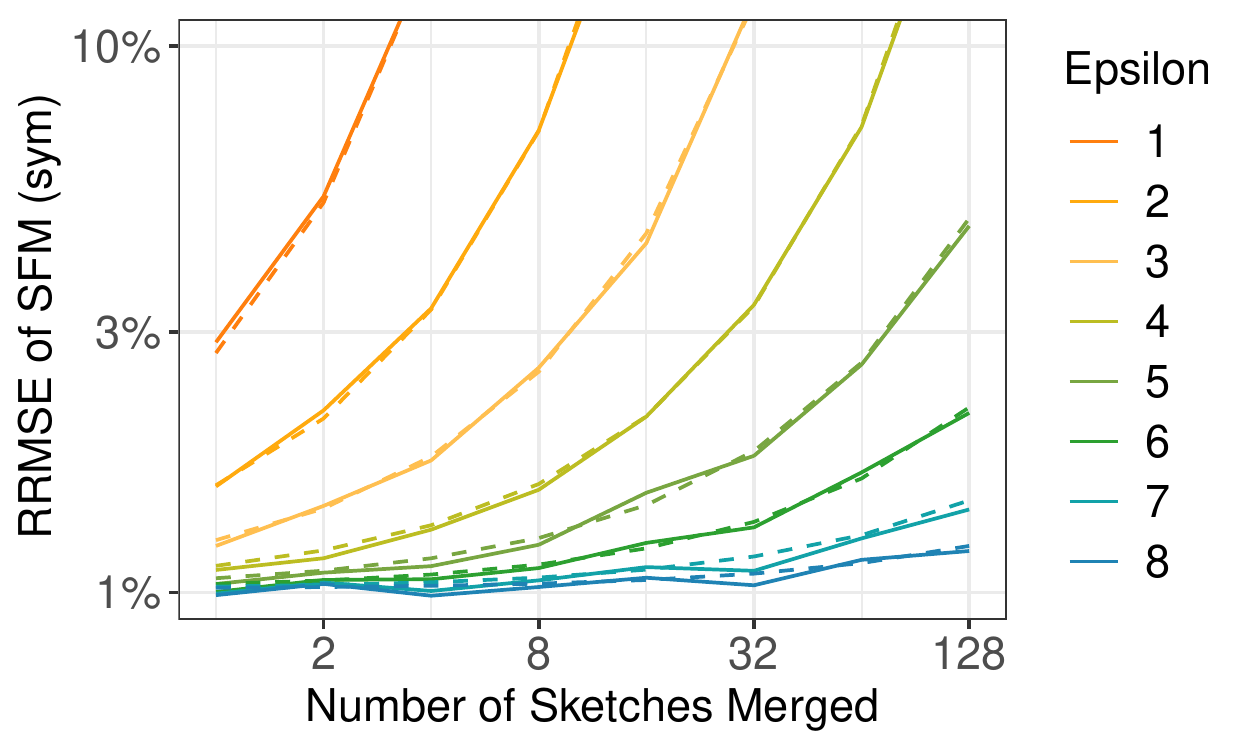}
    \caption{RRMSE at $n=10^6$ after merging a given number of \ussym{} summaries, each with a given privacy budget $\eps$. Dashed lines show estimated relative error.}
    \label{fig:merges}
\end{figure}

Figure~\ref{fig:error-by-n-all} compares the accuracy of the four private methods on synthetic data as the cardinality $n$ ranges from $10^2$ to $10^6$ given a fixed privacy budget of $\eps = 1$. 
For large cardinalities, RRMSE tends toward a fixed constant for each method. Our methods have error that is an order of magnitude better than \citet{pagh2020efficient} even after we tighten their privacy analysis.
For small cardinalities, the relative error increases as the cardinality decreases, which is expected for differentially private methods.
Figure~\ref{fig:error-by-n} compares the accuracy for multiple values of $\eps$ but only for the best sketch, \ussym{}. It also shows that the RRMSE stabilizes as $n \to \infty$ regardless of the choice of $\eps$. 
In contrast to DP methods, the SFM summary with infinite privacy budget, which is the same as PCSA with the MDL estimator of \citet{lang2017back}, yields especially accurate estimates at small $n$.
Figure~\ref{fig:error-by-n} further shows that the estimated relative error $\estse/n$ (Eq.~\ref{eq:fisher-error}) is an upper bound on the empirical error and yields a good estimate of RRMSE, especially for large cardinalities or small $\eps$.

Figure~\ref{fig:merges} demonstrates the tradeoff between merging and privacy in SFM summaries at large cardinality ($n=10^6$). Merge operations result in an accumulation of noise, requiring the use of larger privacy budgets to accommodate greater merge counts. The estimated relative error here is calculated according to Remark~\ref{remark:k-way-merge} and once again closely matches empirical error.

We also compare private methods against the real-world BitcoinHeist data over a variety of privacy budgets $\eps$, ranging from 0.25 to 4. 
  The left panel of Figure~\ref{fig:head-to-head} shows the relative efficiency of \ussym{} as compared with the other private methods. \ussym{} is uniformly more efficient than the PS estimators by at least an order of magnitude. The $100\times$ better efficiency compared to \psloose{} implies \ussym{} would require just 1\% of the space to achieve the same error.
Moreover, \ussym{} outperforms \usxor{} with relative efficiency tending toward 4 for larger $\eps$, indicating that for larger privacy budgets, the randomized-merge \ussym{} can achieve comparable accuracy to \usxor{} in as little as one fourth the space. The estimation error from this experiment is depicted in absolute terms in the right panel, where we again see that the estimated relative error for SFM is a good approximation for RRMSE.

\begin{figure}
    \centering
    \includegraphics[width=.6\textwidth]{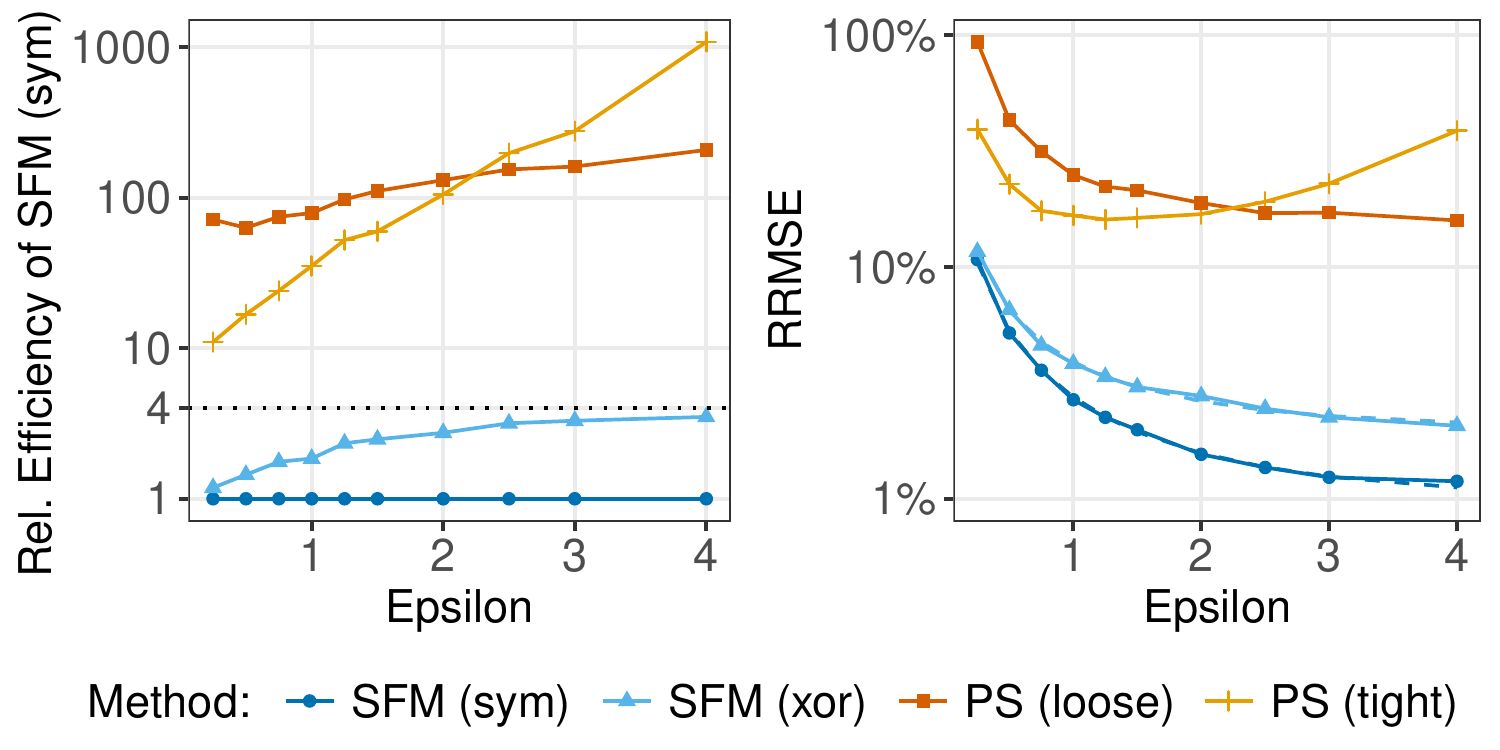}
    \caption{
    \textit{(left)} Relative efficiency of \ussym{} vs. other methods on BitcoinHeist ($n \approx 2.6\mathrm{M}$). SFM is far more efficient than PS. \ussym{} is 4x more efficient than \usxor{} for larger privacy budgets. \textit{(right)} RRMSE vs. $\eps$. For SFM, dashed lines show estimated relative error match the true error.}
    \label{fig:head-to-head}
\end{figure}

\begin{figure}
    \centering
    \includegraphics[width=.5\textwidth]{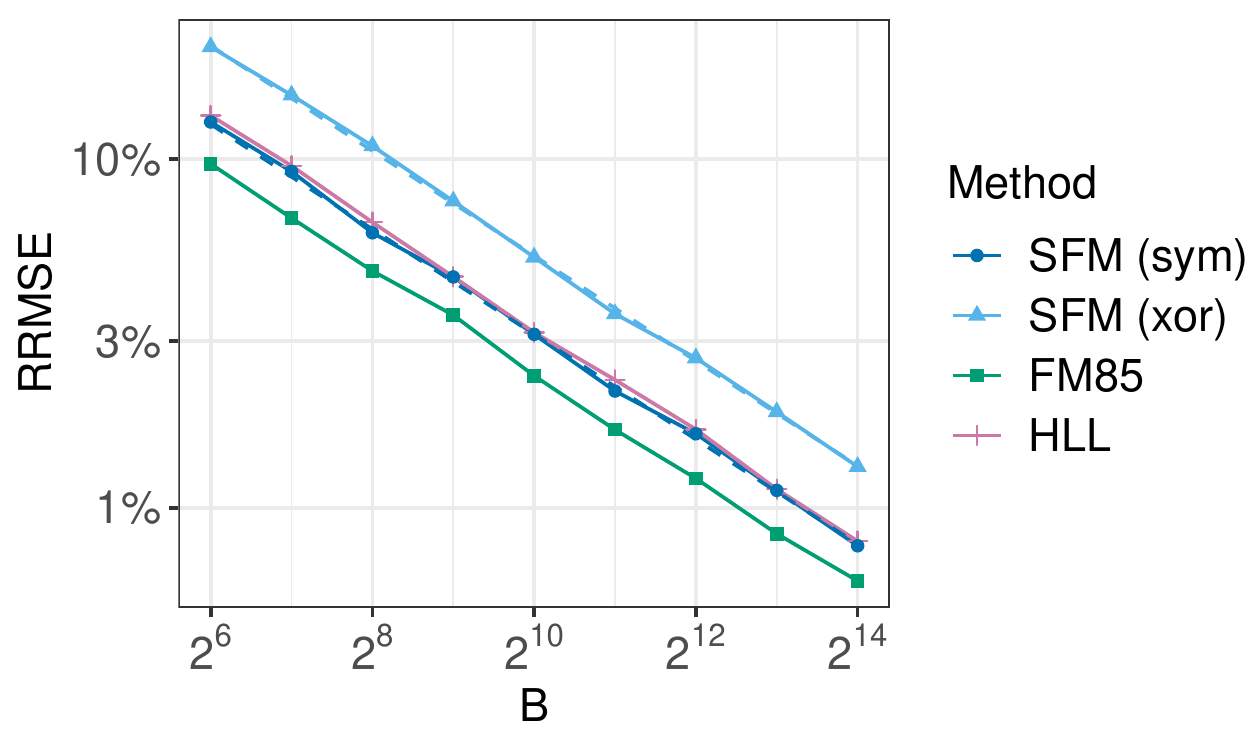}
    \caption{RRMSE at $n = 10^6$ vs. bucket count $B$ for private sketches at $\eps = 2$ vs. common non-private alternatives (on log--log axes). For SFM, dashed lines show estimated relative error. For all methods, RRMSE scales with $B^{-0.5}$.}
    \label{fig:buckets}
\end{figure}

Finally, we compare SFM to popular non-private alternatives
and show that error similarly decreases with the bucket count. Using synthetic data with cardinality $n = 10^6$, we construct sketches of varying bucket count $B$, using a privacy budget of $\eps = 2$ for SFM. Figure~\ref{fig:buckets} shows RRMSE as a function of $B$ for each method. Like the familiar non-private distinct counting sketches, our RRMSE decreases with $B^{-0.5}$. Thus, like non-private sketches, the RRMSE of our DP summaries can be easily characterized by a simple formula $c_{\eps} / \sqrt{B}$ 
at large cardinalities, where $c_{\eps}$ is a constant specific to a method and privacy budget.

\section{Discussion and Conclusion}
\label{sec:discussion}
The Sketch-Flip-Merge summaries demonstrate dramatic improvement over the current state-of-the-art mergeable and differentially private distinct-count sketches. This is achieved through novel merge algorithms (Theorem~\ref{thm:sym-rr-merge} and Section~\ref{sec:general-boolean-ops}), asymptotically optimal estimation (Section~\ref{sec:sketches}), and an improved privacy analysis (Corollary~\ref{cor:equivalent-eps}).

An important limitation in mergeable private summaries is the inherent tension between privacy and mergeability. While both are attainable, repeated merging in the private setting degrades accuracy. This tradeoff, argued in the general distinct-count setting by \citet{desfontaines2018cardinality}, is explicitly quantified for SFM in Remark~\ref{remark:k-way-merge} and Figure~\ref{fig:merges}.

Finally, we note the generality of some of our findings.
In particular, our methods for aggregating noisy binary data provide fundamental machinery and a quantification of the noise-compounding effects of bitwise operations under randomized response that apply to a wide array of problems, particularly in the privacy-preserving space.

\pagebreak


{
    \bibliographystyle{plainnat}
    \bibliography{sources}
}

\pagebreak


\appendix

\section{Proofs of Main Results}
\label{sec:proofs}

\subsection{Additional Notation}

In the proofs to follow, we use the following notation. $\mathbf 1_k$ is the column vector of $k$ ones, and $I_k$ is the $k \times k$ identity matrix (with subscripts omitted when dimensions are clear).  We denote by $e_i$ the $i$-th elementary basis vector, i.e., the vector whose entries are all 0, except at position $i$, where the entry is 1. We use $\langle \cdot, \cdot \rangle$ to denote the inner product between two vectors.

\subsection{Proofs of Results in Section~\ref{sec:privacy}}

\begin{proof}[Proof of Theorem~\ref{thm:dp-bit-flip}]
The proof follows along the same lines as \citet[Theorem 2]{rappor}. 
We first prove that $\F_{p, q}$ is $\eps$-DP if and only if $p \in (0, 1)$ and $\max \{ p/q, (1-q)/(1-p) \} \leq e^\eps$. For this to hold, we must satisfy
\[
\Pr(\F_{p,q}(x) = y) \leq e^\eps \, \Pr(\F_{p,q}(1-x) = y)
\]
for all $x, y \in \{0, 1\}$. 
If $p = q$,  $\F_{p,p}$ is a mechanism that ignores its inputs and outputs a randomly chosen bit value and so this holds trivially. 
Assuming $p \neq q$, this holds if
\[
\Pr(\F_{p,q}(1-x) = y) > 0 \quad \text{and} \quad \frac{\Pr(\F_{p,q}(x) = y)}{\Pr(\F_{p,q}(1-x) = y)} \leq e^\eps .
\]
The first condition is equivalent to the requirement that $p, q \in (0, 1)$, and the second condition is equivalent to $\max \{ p/q, (1-q)/(1-p) \} \leq e^\eps$.

For $\M_{p,q}$, note that the entries of $\M_{p,q}$ are independent by construction, so for $x, x' \in \{0, 1\}^d$ differing only on one bit $x_j = 1 - x'_j$ and some $y \in \{0, 1\}^d$, we have:
\[
\frac{\Pr(\M_{p,q}(x) = y)}{\Pr(\M_{p,q}(x') = y)} = \frac{\prod_i \Pr(\F_{p,q}(x_i) = y_i)}{\prod_i \Pr(\F_{p,q}(x'_i) = y_i)} = \frac{\Pr(\F_{p,q}(x_j) = y_j)}{\Pr(\F_{p,q}(1 - x_j) = y_j)}.
\]
As shown above, this quantity is bounded by $e^\eps$ under the stated conditions on $p$ and $q$.
\end{proof}

Theorem~\ref{thm:dp-bit-flip} is used below to demonstrate that both $\Mxor_\eps$ and $\Msym_\eps$ satisfy $\eps$-DP via Corollaries~\ref{thm:xor-rr-dp} and \ref{thm:sym-rr-dp}.

\subsection{Proofs of Results in Section~\ref{sec:merging-deterministic}}

Before proving Theorem~\ref{thm:unique-deterministic-merge}, we provide a helpful result:

\begin{fact}
\label{thm:bernoulli-ops}
Let $X \sim \bernoulli(p), Y \sim \bernoulli(q)$ be independent. Then:
\begin{enumerate}
    \item $X \land Y \sim \bernoulli(pq)$.
    \item $X \lor Y \sim \bernoulli(p(1-q) + q(1-p) + pq)$.
    \item $X \lxor Y \sim \bernoulli(p(1-q) + q(1-p))$. Moreover, if $p = \frac 12$ or $q = \frac 12$, then $X \lxor Y \sim \bernoulli(\frac 12)$.
\end{enumerate}
\end{fact}
\begin{proof}
\[
\begin{aligned}
\Pr(X \land Y = 1) &= \Pr(X=1, Y=1) \\
&= pq, \\
\Pr(X \lxor Y = 1) &= \Pr(X = 1, Y = 0) + \Pr(X = 0, Y = 1) \\
&= p(1-q) + q(1-p), \\
\Pr(X \lor Y = 1) &= \Pr(X \land Y  = 1) + \Pr(X \lxor Y = 1) .
\end{aligned}
\]
\end{proof}

\begin{proof}[Proof of Theorem~\ref{thm:unique-deterministic-merge}]
We begin with some simple necessary conditions for (1)--(3) to hold. 
From Theorem~\ref{thm:dp-bit-flip}, we know that $p_1, p_2, q_1, q_2$ must lie in $(0, 1)$. Another necessary condition is that $p_1 \neq q_1$ and $p_2 \neq q_2$, as otherwise (3) is violated. For this reason, we can assume $p_1 \neq q_1$ and $p_2 \neq q_2$.

Modulo negation, there exist four symmetric operations $\{0, 1\}^2 \to \{0, 1\}$: and ($\land$), or ($\lor$), xor ($\lxor$), and the trivial operator that maps all inputs to 0. We will now rule out the operators other than $\lxor$.

\textit{(Trivial Operator)} Let $\circ$ denote the operator $x \circ y = 0$. Then $f_1(x) \circ f_2(y) = 0$ for any $x, y$. If (2) holds, then $f_3(0) = f_3(1) = 0$, violating (3).

\textit{(And)} Let $\circ = \land$, and assume (2) holds. Then we must have:
\[
f_1(0) \land f_2(1) \overset D= f_3(1) \overset D= f_1(1) \land f_2(1) .
\]
By Fact~\ref{thm:bernoulli-ops}, this implies that:
\[
q_1 p_2 = p_1 p_2 .
\]
Since we assumed $p_2 \neq 0$, this implies that $p_1 = q_1$, in contradiction to our assumption.

\textit{(Or)} Let $\circ = \lor$, and assume (2) holds. Then we must have:
\[
f_1(0) \lor f_2(1) \overset D= f_3(1) \overset D= f_1(1) \lor f_2(1) .
\]
By Fact~\ref{thm:bernoulli-ops}, this implies that:
\[
p_1 (1-q_2) + q_2 (1-p_1) + p_1 q_2 = p_1 (1-p_2) + p_2 (1-p_1) + p_1 p_2.
\]
After rearranging terms, we obtain:
\[
0 = (1-p_1)(p_2 - q_2) .
\]
Since we assumed $p_1 \neq 1$, this implies that $p_2 = q_2$, in contradiction to our assumption.

\textit{(Xor)} 
Now we will show that when $\circ = \lxor$, we must have $p_1 = p_2 = \frac 12$. Assuming condition (2) holds, we have:
\[
f_1(0) \lxor f_2(1) \overset D= f_3(1) \overset D= f_1(1) \lxor f_2(1),
\]
which implies (by Fact~\ref{thm:bernoulli-ops}):
\[
q_1 (1-p_2) + p_2 (1 - q_1) = p_1 (1 - p_2) + p_2 (1 - p_1) .
\]
Rearranging terms yields:
\[
p_2 (p_1 - q_1) = (1 - p_2) (p_1 - q_1) .
\]
Since we assumed $p_1 \neq q_1$, we must have $p_2 = \frac 12$. A similar argument shows $p_1 = \frac 12$.
\end{proof}

\begin{proof}[Proof of Lemma~\ref{thm:xor-rr-dp}]
We need to show that $p/q \leq e^\eps$ and $(1-q)/(1-p) \leq e^\eps$
for $p=1/2$ and $q=1/(2e^\eps)$.
Clearly, $p/q = e^\eps$, satisfying the first component. Next, consider the expression
\[
\begin{aligned}
e^\eps - \frac{1-q}{1-p} &= e^\eps - \frac{1 - \frac12 e^{-\eps}}{\frac 12} \\
    &= e^\eps - (2 - e^{-\eps}) \\
    &= e^\eps + e^{-\eps} - 2 \\
    &= (e^{\eps/2} - e^{-\eps/2})^2 \\
    &\geq 0.
\end{aligned}
\]
Therefore, $(1-q)/(1-p) \leq e^\eps$ as required.
\end{proof}

\begin{proof}[Proof of Theorem~\ref{thm:xor-rr-merge}]
Since the entries of $\Mxor(\cdot)$ are independent, it suffices to show this holds for $\Mxor$ applied to arbitrary $x_i, y_i \in \{0, 1\}$. Observe that if $x_i = 1$ or $y_i = 1$, then $x_i \lor y_i = 1$, and so $\Mxor_{\eps^*}(x_i \lor y_i) \sim \bernoulli(\frac 12)$. On the other hand, since we know $x_i = 1$ or $y_i = 1$, then $\Mxor_{\eps_1}(x_i) \sim \bernoulli(\frac 12)$ or $\Mxor_{\eps_2}(y_i) \sim \bernoulli(\frac 12)$, and so by Fact~\ref{thm:bernoulli-ops}, $\Mxor_{\eps_1}(x_i) \lxor \Mxor_{\eps_2}(y_i) \sim \bernoulli(\frac 12)$.

Thus all that remains to show is that $\Mxor_{\eps^*}(x \lor y) \overset D= \Mxor_{\eps_1}(x_i) \lxor \Mxor_{\eps_2}(y_i)$ when $x_i = y_i = 0$. In this case, $\Mxor_{\eps_1}(x_i) \sim \bernoulli(\frac 12 e^{-\eps_1})$, and $\Mxor_{\eps_2}(y_i) \sim \bernoulli(\frac 12 e^{-\eps_2})$. By Fact~\ref{thm:bernoulli-ops}, we have that:
\[
\Mxor_{\eps_1}(x_i) \lxor \Mxor_{\eps_2}(y_i) \sim \bernoulli(q^*),
\]
where
\[
\begin{aligned}
q^* &=\frac 12 e^{-\eps_1} \left(1 - \frac 12 e^{-\eps_2} \right) + \frac 12 e^{-\eps_2} \left(1 - \frac 12 e^{-\eps_1} \right) \\
&= \frac 12 \left( e^{-\eps_1} + e^{-\eps_2} - e^{-(\eps_1 + \eps_2)} \right) \\
&= \frac 12 \exp\left( - \left[ \underbrace{-\log\left( e^{-\eps_1} + e^{-\eps_2} - e^{-(\eps_1 + \eps_2)} \right)}_{\eps^*} \right] \right) \\
&= \frac 12 e^{-\eps^*}.
\end{aligned}
\]
Finally, since $x_i \lor y_i = 0$, we have that $\Mxor_{\eps^*}(x_i \lor y_i) \sim \bernoulli(\frac 12 e^{-\eps^*})$.
\end{proof}

\subsection{Proofs of Results in Section~\ref{sec:merging-randomized}}

\begin{proof}[Proof of Lemma~\ref{thm:sym-rr-dp}]
Since $p = e^{\eps}/(e^\eps + 1)$ and $q = 1-p = 1/(e^{\eps} + 1)$, we have $p/q = (1-q)/(1-p) = e^\eps$, and so the result follows immediately from Theorem~\ref{thm:dp-bit-flip}.
\end{proof}

We now state a more general form of Theorem~\ref{thm:sym-rr-merge} for proof. Where Theorem~\ref{thm:sym-rr-merge} gives a randomized merge for 2 bits, which may be invoked repeatedly to merge $k > 2$ bits, Theorem~\ref{thm:sym-merge-general} considers a simultaneous merge of $k \geq 2$ bits. Beyond simply serving to prove the original pairwise theorem, this generalization shows that nothing is gained in $\eps^*$ (i.e., the noise level of the final sketch) by simultaneously merging $k$ bits vs. performing repeated pairwise merges.

\begin{theorem}
\label{thm:sym-merge-general}
Fix an integer $k \geq 2$. For $i \in [k]$, assume $\eps_i > 0$. Let 
\[
q(\eps) = \frac{1}{e^\eps + 1},  \ K_{\eps} = \begin{bmatrix} 1 - q(\eps) & q(\eps) \\ q(\eps) & 1 - q(\eps) \end{bmatrix},
\]
\[
\eps^* = -\log \left(1 - \prod_{i=1}^k (1 - e^{-\eps_i}) \right), \ q^* = q(\eps^*),
\]
and let $v^* \in \R^{2^k}$ be the vector whose first entry is $q^*$ with all other entries $1 - q^*$. Let
\[
t = (t_{\dots01}, t_{\dots10}, t_{\dots11}, \dots)^T = (K_{\eps_1}^{-1} \otimes \dots \otimes K_{\eps_k}^{-1}) v^*,
\]
\[
g(x_1, \dots, x_k) \sim \bernoulli(t_{x_1 \dots x_k}).
\]
Then $g(\Msym_{\eps_1}(x_1), \dots, \Msym_{\eps_k}(x_k)) \overset D= \Msym_{\eps^*}(x_1 \lor \dots \lor x_k)$.
\end{theorem}

Before proving Theorem~\ref{thm:sym-merge-general}, we provide another fact that will be used in the proof.

\begin{fact}
\label{thm:row-sums-one}
Let $A_{m_A \times n_A}, B_{m_B \times n_B}$ be matrices satisfying $A \mathbf 1 = \mathbf 1$ and $B \mathbf 1 = \mathbf 1$. Then $(AB) \mathbf 1 = \mathbf 1$ and $(A \otimes B) \mathbf 1 = \mathbf 1$. Additionally, if $A^{-1}$ exists, then $A^{-1} \mathbf 1 = \mathbf 1$.
\end{fact}
\begin{proof}
Since $A \mathbf 1 = \mathbf 1$ and $B \mathbf 1 = \mathbf 1$, we must have $(AB) \mathbf 1 = A (B \mathbf 1) = A \mathbf 1 = \mathbf 1$.

Next, write $\mathbf 1_{m_A m_B} = \mathbf 1_{m_A} \otimes \mathbf 1_{m_B}$. Then
\[
\begin{aligned}
(A \otimes B) \mathbf 1_{m_A m_B} &= (A \otimes B) (\mathbf 1_{m_A} \otimes \mathbf 1_{m_B}) \\
    &= (A \mathbf 1_{m_A}) \otimes (B \mathbf 1_{m_B}) \\
    &= \mathbf 1_{m_A} \otimes \mathbf 1_{m_B} \\
    &= \mathbf 1_{m_A m_B}.
\end{aligned}
\]

Finally, in the case when $A$ is invertible, $A^{-1} A = I$, and so we must have $A^{-1} \mathbf 1 = A^{-1} A \mathbf 1 = I \mathbf 1 = \mathbf 1$.
\end{proof}

\begin{proof}[Proof of Theorem~\ref{thm:sym-merge-general}]
Since all operations are performed bitwise and independently, assume without loss of generality that the bit vectors $x_i$ are scalar, i.e., $x_1, \dots, x_k \in \{0, 1\}$.

The fundamental idea is to model the chain of operations performed on the bits $x_i$ as a Markov chain. 
This involves three different types of transition probability matrices. 
The matrices $K_\eps$ defined in the theorem statement map the state space of a single bit in $\{0, 1\}$ to another bit in $\{0, 1\}$ via the application of $\Msym_\eps$. 
Next, we define $K^{\operatorname{or}}$ to be the $2^k \times 2$ matrix mapping $k$ bits in $\{0, 1\}^k$ to a single bit in $\{0, 1\}$ via an \textit{or} operation. Finally, we define $K^{\merge}$ to be the $2^k \times 2$ matrix corresponding to our desired merge operation, which maps $\{0, 1\}^k$ to $\{0, 1\}$.

Since the bit-flipping operations of $K_\eps$ are performed independently, the matrix $K_{\eps_1} \otimes \dots \otimes K_{\eps_k}$ represents the $2^k \times 2^k$ matrix jointly mapping the state space of the original bits $\{x_i\}_{i=1}^k$ to $\{\Msym_{\eps_i}(x_i)\}_{i=1}^k$. Thus we wish to solve:
\begin{equation}
\label{eq:transition-matrix-eq}
(K_{\eps_1} \otimes \dots \otimes K_{\eps_k}) K^{\operatorname{merge}}_{\eps'} = K^{\operatorname{or}} K_{\eps'},
\end{equation}
where $\eps'$ is a free parameter that we will fix, and $K^{\merge}_{\eps'}$ is the unknown quantity. We proceed by solving the matrix equation above, finding the maximum $\eps'$ for which $K^{\merge}_{\eps'}$ represents a valid transition probability matrix.

Let $q_i = q(\eps_i), q' = q(\eps')$. We note that $K_{\eps_i}$ is invertible and write $K_{\eps_i}^{-1}$ as follows:
\[
K_{\eps_i}^{-1} = \frac{1}{1 - 2 q_i} \begin{bmatrix}
    1 - q_i & -q_i \\
    -q_i & 1 - q_i
\end{bmatrix}.
\]
So we may solve our matrix equation (\ref{eq:transition-matrix-eq}) by left-multiplication of $(K_{\eps_1} \otimes \dots \otimes K_{\eps_k})^{-1}$:
\[
K^{\merge}_{\eps'} = (K_{\eps_1}^{-1} \otimes \dots \otimes K_{\eps_k}^{-1}) K^{\operatorname{or}} K_{\eps'} .
\]
The first column of $K^{\operatorname{or}} K_{\eps'}$ is equal to
\[
w' = (1-q', q', \dots, q')^T .
\]
It follows from Fact~\ref{thm:row-sums-one} that $K^{\merge}_{\eps'} \mathbf 1 = \mathbf 1$ and therefore $K^{\merge}_{\eps'}$ is stochastic if and only if:
\[
u' = (K_{\eps_1}^{-1} \otimes \dots \otimes K_{\eps_k}^{-1}) \, w' \in [0, 1]^{2^k} .
\]
We may write $u'_i$, the $i$-th entry of $u'$, as the inner product of $w'$ with the $i$-th row of $(K_{\eps_1}^{-1} \otimes \dots \otimes K_{\eps_k}^{-1})$. We denote by $r_i$ this row vector. Writing $w' = q' \mathbf 1 + (1 - 2q') e_1$, $u'_i$ may be written:
\[
\begin{aligned}
u'_i &= \langle r_i, \; w' \rangle \\
    &= \langle r_i, \; q' \mathbf 1 + (1 - 2q') e_1 \rangle \\
    &= q' \langle r_i, \; \mathbf 1 \rangle + (1 - 2q') \langle r_i, e_1 \rangle \\
    &= q' + (1 - 2q') (K_{\eps_1}^{-1} \otimes \dots \otimes K_{\eps_k}^{-1})_{i1},
\end{aligned}
\]
where the final equality comes from the fact that $\langle r_i, \mathbf 1 \rangle = 1$ (Fact~\ref{thm:row-sums-one}) and that $\langle r_i, e_1 \rangle$ is equal to the $(i, 1)$-th entry of $K_{\eps_1}^{-1} \otimes \dots \otimes K_{\eps_k}^{-1}$.

Since each $u'_i$ entry must be in $[0, 1]$, each entry defines a constraint on $q'$. In particular, since these values are affine functions of $q'$, each constraint corresponds to an interval. We can see that $q' = \frac 12$ is valid for each of these constraints, as $u'_i = \frac 12$ when $q' = \frac 12$ for all $i$. Thus it suffices to find a lower bound for $q'$ using these constraints.

We ask next which values of $u'_i$ are most extreme. For any fixed choice of $q' \in [0, \frac 12]$, we obtain the largest entry $u'_i$ where $(K_{\eps_1}^{-1} \otimes \dots \otimes K_{\eps_k}^{-1})_{i1}$ is maximized and the smallest where that same value is minimized. In particular:
\[
\max_i (K_{\eps_1}^{-1} \otimes \dots \otimes K_{\eps_k}^{-1})_{i1} = \left( \prod_i (1 - 2q_i)^{-1} \right) \prod_i (1 - q_i),
\]
since $1 - q_i > q_i > 0$ for all $i$. (The constant $\prod_i (1 - 2q_i)^{-1}$ appears in all entries.) Similarly, we can see that:
\[
\min_i (K_{\eps_1}^{-1} \otimes \dots \otimes K_{\eps_k}^{-1})_{i1} = \left( \prod_i (1 - 2q_i)^{-1} \right) (-q_j) \prod_{i \neq j} (1 - q_i),
\]
where $j = \arg \max_j q_j$, since this yields the most extreme negative term.

It suffices to constrain the two most extreme entries of $u'$ to $[0, 1]$. The constraints defined by these two entries are:
\[
q' + (1 - 2q') \underbrace{\frac{\prod_i (1 - q_i)}{\prod_i (1 - 2q_i)}}_{=:c_1} \leq 1 \iff q' \overset{(a)}{\geq} \frac{c_1 - 1}{2c_1 - 1},
\]
and
\[
q' + (1 - 2q') \underbrace{(-q_j) \frac{\prod_{i \neq j}(1 - q_i)}{\prod_i(1 - 2q_i)}}_{=:-c_2} \geq 0 \iff q' \overset{(b)}{\geq} \frac{c_2}{1 + 2c_2} .
\]

We note that $c_1 > 1$ and $c_2 > 0$. Moreover, comparing $c_1$ and $c_2$, we find that $c_1 > c_2 + 1$, as:
\[
\begin{aligned}
c_1 - (c_2 + 1) &= \frac{\prod_i (1 - q_i)}{\prod_i (1 - 2q_i)} - \frac{q_j \prod_{i \neq j} (1 - q_i) + \prod_i (1 - 2q_i)}{\prod_i (1 - 2q_i)} \\
    &= \frac{(1 - q_j) \prod_{i \neq j} (1 - q_i)}{\prod_i (1 - 2q_i)} - \frac{q_j \prod_{i \neq j} (1 - q_i) + \prod_i (1 - 2q_i)}{\prod_i (1 - 2q_i)} \\
    &= \frac{(1 - 2q_j) \prod_{i \neq j} (1 - q_i) - \prod_i (1 - 2q_i)}{\prod_i (1 - 2q_i)} \\
    &> 0.
\end{aligned}
\]
So it follows that:
\[
\begin{aligned}
\frac{c_1 - 1}{2c_1 - 1} - \frac{c_2}{1 + 2c_2} &= \frac{(c_1 - 1)(1 + 2 c_2) - c_2 (2c_1 - 1)}{(2c_1 - 1)(1 + 2c_2)} \\
    &= \frac{c_2 - (c_1 + 1)}{(2c_1 - 1)(1 + 2c_2)} \\
    &> 0,
\end{aligned}
\]
which indicates that constraint $(a)$ implies constraint $(b)$. 

We denote by $q^*$ the minimal $q'$ allowed under constraint $(a)$:
\[
q^* = \frac{c_1 - 1}{2c_1 - 1} = \frac{\prod_i \frac{1-q_i}{1-2q_i} - 1}{2\prod_i \frac{1-q_i}{1-2q_i} - 1} .
\]
Putting this in terms of $\eps$, note that $1 - q_i = \frac{e^{\eps_i}}{e^{\eps_i}+1}$, $1 - 2q_i = \frac{e^{\eps_i} - 1}{e^{\eps_i}+1}$, so:
\[
\frac{1-q_i}{1-2q_1} = \frac{e^{\eps_i}}{e^{\eps_i}-1} = (1 - e^{-\eps_i})^{-1},
\]
and hence:
\[
q^* = \frac{\prod_i (1 - e^{-\eps_i})^{-1} - 1}{2\prod_i (1 - e^{-\eps_i})^{-1} - 1} = \frac{1 - \prod_i (1 - e^{-\eps_i})}{2 - \prod_i (1 - e^{-\eps_i})} ,
\]
which gives a final $\eps^*$ of:
\[
\begin{aligned}
\eps^* &= q^{-1}(q^*) \\
    &= \log \left( (q^*)^{-1} - 1 \right) \\
    &= \log \left( \frac{2 - \prod_i (1 - e^{-\eps_i})}{1 - \prod_i (1 - e^{-\eps_i})} - 1\right) \\
    &= \log \left( \frac{1}{1 - \prod_i (1 - e^{-\eps_i})} \right) \\
    &= -\log \left( 1 - \prod_{i=1}^k (1 - e^{-\eps_i}) \right) .
\end{aligned}
\]
Finally, to translate from transition probability matrices back to the theorem statement, note that $K^{\merge}_{\eps^*} = (K_{\eps_1}^{-1} \otimes \dots \otimes K_{\eps_k}^{-1}) K^{\operatorname{or}} K_{\eps^*}$ maps the $\{0, 1\}^k$ state space to $\{0, 1\}$---i.e., the $2^k$ possible inputs map to Bernoulli random variables with probabilities taken from the second column of $K^{\merge}_{\eps^*}$. It follows from the preceding discussion that this vector is precisely $(K_{\eps_1}^{-1} \otimes \dots \otimes K_{\eps_k}^{-1}) v^*$.
\end{proof}

\subsection{Proofs of Results for Section~\ref{sec:sketches}}
\label{sec:sketch-proofs}

\begin{proof}[Proof of Theorem~\ref{thm:shape-of-loglik}]
Let $f(\hat n) = \E[\ell(\hat n; T)]$. 
We will use the notation $\E_n[\cdot]$ to denote expectation under a cardinality of $n$, while abusing notation with $\E_{\hat n}[\cdot]$ to denote the equivalent quantity with $\hat n$ replacing $n$, as below:
\[
\begin{aligned}
\E_n[T_{ij}] &= \Pr(T_{ij} = 1) \\
    &= p (1 - \gamma_j^n) + q \gamma_j^n \\
    &= p - (p-q) \gamma_j^n \\
\E_{\hat n}[T_{ij}] &= p - (p-q) \gamma_j^{\hat n} .
\end{aligned}
\]
(Note that $\E_{\hat n}[\cdot]$ is not truly an expectation, since when $\hat n$ is non-integer, the distribution of $T$ is not defined.) Observe that:
\[
\begin{aligned}
\E_n[f'(\hat n)] &= B \sum_{j=1}^P (1 - \E_n[T_{ij}]) (p-q)\gamma_j^{\hat n} \log(\gamma_j) (1 - \E_{\hat n}[T_{ij}])^{-1} \\
&\quad\quad - B \sum_{j=1}^P \E_n[T_{ij}] (p-q)\gamma_j^{\hat n} \log(\gamma_j) (\E_{\hat n}[T_{ij}])^{-1} \\
&= B(p-q) \sum_{j=1}^P \phi'_j(\hat n) ,
\end{aligned}
\]
where
\[
\phi'_j(\hat n) = \gamma_j^{\hat n} \log(\gamma_j) \left( \frac{1 - \E_n[T_{ij}]}{1 - \E_{\hat n}[T_{ij}]} - \frac{\E_n[T_{ij}]}{\E_{\hat n}[T_{ij}]} \right) .
\]
As expected, this equals zero when $\hat n = n$. Moreover, it is strictly positive for $\hat n < n$ and strictly negative for $\hat n > n$. Thus the same properties hold for $f'(\hat n)$, and so $n$ is the global maximizer of $f$.

By similar logic, we may write
\[
\E_n[f''(\hat n)] = B(p-q) \sum_{j=1}^P \phi''_j(\hat n) ,
\]
where
\[
\phi''_j(\hat n) = (\log\gamma_j)^2 \gamma_j^{\hat n} \left( (1-p) \frac{1 - \E_n[T_{ij}]}{(1 - \E_{\hat n}[T_{ij}])^2} - p \frac{\E_n[T_{ij}]}{(\E_{\hat n}[T_{ij}])^2} \right) .
\]
Although $\phi''_j(\hat n) > 0$ for sufficiently large $\hat n$, note that the parenthetical quantity is monotonically increasing in $\hat n$. Moreover, we know $\phi''_j(n) < 0$ since $\phi'_j(n)$ corresponds to a maximum. Thus the parenthetical (and indeed all of $\phi''_j(\hat n)$) must be negative for $\hat n \leq n$. Since this is true for all $j$, it follows that $f''(\hat n) < 0$ for $\hat n \leq n$.
\end{proof}

\begin{lemma}
\label{lem:constant update probability}
Consider a bucket in a PCSA summary with $v$ items where the bucket has $P=\infty$ bits. The probability that a new item allocated to the bucket modifies the bucket is bounded by $c / v$ for all $v > v_0$ 
for some constants $c, v_0$.
\end{lemma}
\begin{proof}
The probability a bucket containing $v$ items is modified by a new item allocated to the bucket is $\sum_{\ell=1}^\infty 2^{-\ell} \left(1-2^{-\ell}\right)^v$.
Split this sum into the ranges $\ell \in \ical_1 := (0, \log_2 v - \log_2 \log_2 v)$,
$\ell \in \ical_2 := [\log_2 v - \log_2 \log_2 v, \log_2 v)$, $\ell \in \ical_3 := [\log_2 v, \infty)$.
Since $2^{-\ell} \leq 1$ and $(1-2^{-\ell})^v < \exp(-2^{-\ell} v) \leq 1$,
\begin{align*}
\sum_{\ell \in \ical_1} 2^{-\ell} \left(1-2^{-\ell}\right)^v
&\leq \log_2 v \exp(-\log_2 v) = o(1), \\ 
\sum_{\ell \in \ical_2} 2^{-\ell} \left(1-2^{-\ell}\right)^v
&\leq \int_{\log_2 v - \log_2 \log_2 v}^{\log_2 v} 2^{-x} \exp(-2^{-x} v) dx \\
&= \frac{1}{v \log 2} \exp(-v 2^{-x}) \Big|_{\log_2 v - \log_2 \log_2 v}^{\log_2 v} \\
&= \frac{e^{-1}}{v \log(2)} - O(\exp(-\log_2 v)/v), \\
\sum_{\ell \in \ical_3} 2^{-\ell} \left(1-2^{-\ell}\right)^v
& \leq \sum_{\ell \in \ical_3} 2^{-\ell} \leq \frac{1}{v} .
\end{align*}
Summing these components gives the desired result.
\end{proof}

\begin{proof}[Proof of Theorem~\ref{thm:sketch convergence in probability}]
The modified PCSA summary $\overline{S}_n$ can be generated in the following way. Draw a new cardinality $\overline{N} \sim \poisson(n)$.
The first $\min \{n,\overline{N}\}$ items are shared for the regular SFM summary $S_n$ and modified summary $\overline{S}_n$. Each of these items are allocated to the same bucket for both summaries. Denote the remaining items by $R = |\overline{N} - n|$. 
The variance of a $\poisson(n)$ gives  $R = O_p(\sqrt{n})$.

Using Theorem~3 in \citet{kolchin1978random} for the asymptotic distribution of the maximum value in a multinomial vector,
the bucket allocated the largest number of remaining items has $O_p(R/B \log B) = O_p(\sqrt{n} / B \log B)$ items. Likewise the bucket with the minimum number of items has $n/B + O_p(\sqrt{n} / B \log B)$ items. 
By Lemma~\ref{lem:constant update probability}, the probability a new item in a bucket will update the bucket's value is $O(1/V_i)$ given $V_i$, the number of items already in the bucket.
Thus, the probability that no bucket is updated by one of the remaining items is 
\begin{align*}
\left(1- \frac{O(\sqrt{n}/ B \log B)}{n/B + O(\sqrt{n}/ B \log B)}\right)^B + o(1) &= \left(1-\frac{O(\log B)}{\sqrt{n}}\right)^B + o(1)\\
&= (1-o(1/B))^B + o(1) \to 1
\end{align*}
 as $n \to \infty$.
This gives that $\Pr(\overline{S}_n = S_n) \to 1$ 
as $n \to \infty$ and the true PCSA sketch and the Poissonized one are asymptotically equal.

We can also relate the Poissonized PCSA sketch to the one whose true likelihood is the composite marginal likelihood.
By Poisson splitting, the entries of $\overline{S}_n$ are independent. Thus, we can couple the entries of $\overline{S}_n$ with those of $\tilde{S}_n$ via the inverse CDF method by using the same underlying $\uniform(0,1)$ random variables.
The probability that an entry in level $j$ is different across the coupled sketches is 
\begin{align*}
\Pr(\tilde{S}_n(i,j) \neq \overline{S}_n(i,j)) &= 
\exp\left(-\frac{n}{B\, 2^j}\right) - \left(1 - \frac{1}{B\, 2^j}\right)^n \\
&= \exp\left(-\frac{n}{B\, 2^j}\right) \left(1 - \exp\left[n \log\left(1-\frac{1}{B 2^j}\right) +\frac{n}{B 2^j} \right]\right) \\
&< \exp\left(-\frac{n}{B\, 2^j}\right) \left(1 - \exp\left[-\frac{n}{B^2 2^{2j+1}} \right]\right) \\
&< \exp\left(-\frac{n}{B\, 2^j}\right) \frac{n}{B^2 2^{2j+1}} .
\end{align*}

Applying a union bound gives and splitting the sum at some positive integer $k$  gives
\begin{align*}
\Pr(\tilde{S}_n \neq \overline{S}_n) &\leq 
B \sum_{j=1}^\infty \exp\left(-\frac{n}{B\, 2^j}\right) \frac{n}{B^2 2^{2j+1}}  \\
&\leq \sum_{j=1}^{k} \exp\left(-\frac{n}{B\, 2^j}\right) \frac{n}{B} +
\sum_{j=k+1}^{\infty} \frac{n}{B 2^{2j+1}}
\\
&\leq \exp\left(-\frac{n}{B\, 2^k}\right) \frac{n k}{B} +
\frac{n}{B 2^{2k}} .
\end{align*}
Take $k = \log_2 \left(\frac{n /B}{2\log ({n}/{B})}  \right) + \delta$
for some $\delta \in [-1/2,1/2)$.
Then the first part
\begin{align*}
\exp\left(-\frac{n}{B\, 2^k}\right) \frac{n k}{B} &=
\frac{nk}{B} \exp\left(-2^{1-\delta} \log(n/B)  \right) \\
&\leq \frac{nk}{B} (n/B)^{-3/2} \to 0 .
\end{align*}
as $n / B \to \infty$.
Likewise, the second part
\begin{align*}
\frac{n}{B 2^{2k}} &= \frac{n}{B} \left(\frac{n /B}{2\log ({n}/{B})}  \right)^{-2} 2^{-2\delta} = \frac{4\log^2 ({n}/{B})}{n/B} 2^{-2\delta} \to 0
\end{align*}
as $n / B \to \infty$.
Thus $\Pr(\tilde{S}_n = \overline{S}_n) \to 1$ as $n \to \infty$ as well.

\end{proof}

\begin{proof}[Proof of Corollary \ref{cor:asymptotic efficiency}]
Since both a PCSA sketch $S_n$ and modified sketch with independent bins $\tilde{S}_n$ are equal with probability going to 1 as $n \to \infty$,
the private SFM sketches $T_n$, $\tilde{T}_n$ obtained by applying the same randomized response noise to them are also equal with probability going to 1.
Let $\hat{\theta}(T_n)$ be some cardinality estimator and $V(\hat{\theta}(T_n))$ denote its asymptotic variance.
Then $\min_{\hat{\theta} \in \Theta} V(\hat{\theta}(T_n)) = \min_{\hat{\theta} \in \Theta} V(\hat{\theta}(\tilde{T}_n))$,
and a cardinality estimator for $T_n$ is asymptotically efficient if and only if it is asymptotically efficient for $\tilde{T}_n$. 
Since the composite likelihood estimator for $T_n$ is the true maximum likelihood estimator for $\tilde{T}_n$, it is asymptotically efficient. 
\end{proof}

\section{Proofs and Results for General Boolean Operations}
\label{sec:and-merging-appendix}

In the main text, we concern ourselves primarily with merge operations under randomized response that emulate the logical operator \textit{or} ($\lor$).
Here we discuss generalizations of these merge operations to other Boolean operations.

\subsection{Boolean Operations under Symmetric Randomized Response}

Recall that a merge for \textit{and} ($\land$) under $\Msym$ was presented in Corollary~\ref{thm:and-sym-rr-merge}.

\begin{proof}[Proof of Corollary~\ref{thm:and-sym-rr-merge}]
We prove this in the more general setting of merging $k$ bits, as in Theorem~\ref{thm:sym-merge-general}. Indeed, proving this is essentially equivalent to Theorem~\ref{thm:sym-merge-general}, except that we must replace $K^{\operatorname{or}}$ with a transition matrix $K^{\operatorname{and}}$ that maps $(x_1, \dots, x_k)$ to $1$ only when $x_1=\dots=x_k=1$. Then for fixed $\eps'$ we have a potential solution:
\[
K^{\merge}_{\eps'} = (K_{\eps_1}^{-1} \otimes \dots \otimes K_{\eps_k}^{-1}) K^{\operatorname{and}} K_{\eps'} .
\]
Once again, we seek the largest $\eps'$ for which $K^{\merge}_{\eps'}$ is a valid transition probability matrix. This time, we will use the second column of $K^{\operatorname{and}}K_{\eps'}$ to determine constraints on $\eps'$. (This is allowable since $K^{\operatorname{and}}K_{\eps'} \mathbf 1 = \mathbf 1$.) We write the second column as:
\[
w' = (q',\dots, q', 1-q')^T = q' \mathbf 1 + (1 - 2q') e_{2k} .
\]
Using $w'$ as in the proof of Theorem~\ref{thm:sym-merge-general}, we obtain the same constraints on $\eps'$. Applying the remainder of the proof of Theorem~\ref{thm:sym-merge-general} yields the final result.
\end{proof}

Next, we demonstrate \textit{xor} ($\lxor$) merging under $\Msym$.

\begin{lemma}
\label{thm:xor-msym}
$\Msym_{\eps_1}(x) \lxor \Msym_{\eps_2}(y) \overset D= \Msym_{\eps^*}(x \lxor y)$ for $\eps^* = \log(1 + e^{\eps_1 + \eps_2}) - \log(e^{\eps_1} + e^{\eps_2})$.
\end{lemma}
\begin{proof}
Let $p_1 = e^{\eps_1} / (e^{\eps_1} + 1), p_2 = e^{\eps_2} / (e^{\eps_2} + 1)$. Using Fact~\ref{thm:bernoulli-ops}, we have that:
\[
\Msym_{\eps_1}(x_i) \lxor \Msym_{\eps_2}(y_i) \sim \bernoulli(\theta_{x_i, y_i}),
\]
where
\[
\begin{aligned}
    \theta_{0,0} = \theta_{1,1} &= p_1 (1 - p_2) + p_2(1 - p_1) \\
    \theta_{0,1} = \theta_{1,0} &= p_1 p_2 + (1 - p_1)(1 - p_2) .
\end{aligned}
\]
Through a bit of algebra, we obtain
\[
\begin{aligned}
\theta_{0,0} = \theta_{1,1} &= \frac{e^{\eps_1}}{(e^{\eps_1} + 1)(e^{\eps_2} + 1)} + \frac{e^{\eps_2}}{(e^{\eps_1} + 1)(e^{\eps_2} + 1)} \\
    &= \frac{e^{\eps_1} + e^{\eps_2}}{e^{\eps_1 + \eps_2} + e^{\eps_1} + e^{\eps_2} + 1} \\
    &= \frac{1}{\frac{1 + e^{\eps_1 + \eps_2}}{e^{\eps_1} + e^{\eps_2}} + 1} ,
\end{aligned}
\]
while
\[
    \theta_{1,0} = \theta_{0,1} = 1 - \theta_{0,0} .
\]
So we have
\[
\Msym_{\eps_1}(x_i) \lxor \Msym_{\eps_2}(y_i) \overset D= \M_{\theta_{1,0}, 1 - \theta_{1,0}}(x_i \lxor y_i) .
\]
Finally, to obtain $\eps^*$, note that $\M_{\theta_{1,0}, 1 - \theta_{1,0}} = \Msym_{\eps^*}$ for
\[
e^{\eps^*} = \frac{1 + e^{\eps_1 + \eps_2}}{e^{\eps_1} + e^{\eps_2}} .
\]
\end{proof}

As our final step in supporting general Boolean operations under $\Msym$, we show that the unary operation \textit{not} ($\lnot$) commutes with $\Msym$.

\begin{lemma}
\label{thm:not-msym}
For any bit vector $x$, we have $\lnot (\M_{p,q}(x)) \overset D= \M_{p,q}(\lnot x)$ if and only if $q = 1-p$. In particular, $\lnot (\Msym_\eps(x)) \overset D= \Msym_\eps(\lnot x)$.
\end{lemma}
\begin{proof}
Note that
\[
\lnot (\M_{p,q}(x_i)) \overset D= \M_{1-p, 1-p}(x_i) \overset D= \M_{1-q,1-p}(\lnot x_i) .
\]
Consequently,
\[
\lnot (\M_{p,q}(x)) \overset D= \M_{p,q}(\lnot x) \iff p = 1 - q \iff q = 1 - p .
\]
\end{proof}

\subsection{Boolean Operations with Deterministic Merging}

In contrast with the results that leveraged randomized merging, we demonstrate that a deterministic merge for a given Boolean operation requires specific choices of randomized response mechanism, precluding the use of general Boolean operations under a single RR mechanism. In particular, we demonstrate that a deterministic \textit{and} ($\land$) merge requires a different privacy mechanism than the \textit{or} ($\lor$) merge described in Theorem~\ref{thm:unique-deterministic-merge}.

\begin{corollary}
\label{thm:and-unique-deterministic-merge}
Let $f_1 = \F_{p_1,q_1}, f_2 = \F_{p_2,q_2}, f_3 = \F_{p_3,q_3}$, and let $\bullet : \{0, 1\}^2 \to \{0, 1\}$ denote a deterministic and symmetric operation. 
The following conditions may only be satisfied simultaneously if $\circ = \lxor$ and $q_1 = q_2 = 1/2$:
\begin{enumerate}
    \item $f_1, f_2$ are (respectively) $\eps_1$-DP and $\eps_2$-DP for $\eps_1, \eps_2 < \infty$.
    \item $f_1(x) \bullet f_2(y) \overset D= f_3(x \land y)$.
    \item $f_i(0) \overset D\neq f_i(1)$ for $i = 1, 2, 3$.
\end{enumerate}
\end{corollary}

\begin{proof}
Note that this theorem statement is identical to that of Theorem~\ref{thm:unique-deterministic-merge} except that condition (2) has changed from using the \textit{or} operation $\lor$ to the \textit{and} operation $\land$ and that the necessary condition is no longer on $p_1, p_2$ but instead $q_1, q_2$.

Consider the original condition (2) of Theorem~\ref{thm:unique-deterministic-merge}. Since this must hold for all $x, y \in \{0, 1\}$, we may alternatively write this condition in terms of $\lnot x$ and $\lnot y$ instead:
\[
\begin{aligned}
\F_{q_1, p_1}(x) \circ \F_{q_2, p_2}(y) &\overset D= \F_{p_1, q_1}(\lnot x) \circ \F_{p_2, q_2}(\lnot y) \\
    & \overset D= f_1(\lnot x) \circ f_2(\lnot y) \\
    & \overset D= f_3((\lnot x) \lor (\lnot y)) \\
    & \overset D= f_3(\lnot (x \land y)) \\
    & \overset D= \F_{q_3, p_3}(x \land y) .
\end{aligned}
\]
From Theorem~\ref{thm:unique-deterministic-merge}, we know that we can only satisfy this condition simultaneously with (1) and (3) if $p_1 = p_2 = 1/2$. Recognizing that this statement is equivalent to condition (2) of the corollary but with the roles of $p_i$ and $q_i$ swapped, it is apparent that to satisfy (1)--(3) of our corollary, we must have $\bullet = \lxor$ and $q_1 = q_2 = 1/2$.
\end{proof}

From here, it follows that no privacy mechanism $\M_{p,q}$ can satisfy the conditions of Theorem~\ref{thm:unique-deterministic-merge} and Corollary~\ref{thm:and-unique-deterministic-merge} simultaneously.

\begin{proof}[Proof of Corollary~\ref{thm:impossible-deterministic-and-or}]
Assume conditions (1)--(4) are satisfied. By Theorem~\ref{thm:unique-deterministic-merge}, we must have $p_1 = p_2 = 1/2$, while Corollary~\ref{thm:and-unique-deterministic-merge} states that $q_1 = q_2 = 1/2$. This results in a contradiction: $f_1(0) \overset D= \bernoulli(1/2) \overset D= f_1(1)$, violating (3).
\end{proof}

\end{document}